\documentclass[journal]{IEEEtran}

\usepackage{amsmath}
\usepackage{amsthm}
\usepackage{arydshln}
\usepackage{amsfonts}
\usepackage{braket}
\usepackage{enumitem}
\usepackage{multirow}
\usepackage{tikz}
\usetikzlibrary{decorations.pathreplacing}
\usepackage{calc}
\usepackage{caption}
\usepackage{setspace}
\usepackage{float}
\usepackage[colorlinks=true,breaklinks=true,urlcolor=black,citecolor=black,linkcolor=black]{hyperref}
\usepackage{multicol}
\usepackage{cuted}
\usepackage{afterpage}
\usepackage{atbegshi}
\usepackage{balance}
\usepackage{flushend}

\newcommand{\ul}[1]{{\underline{#1}}}
\newcommand{\Hs}{\mathcal{H}}
\newcommand{\idmtx}{\mathcal{I}}

\allowdisplaybreaks

\DeclareMathOperator{\Tr}{Tr}
\DeclareMathOperator{\sr}{SR}

\newtheorem{definition}{Definition}
\newtheorem{theorem}{Theorem}
\newtheorem{lemma}{Lemma}
\newtheorem{corollary}{Corollary}

\newcommand{\A}{\mathbb{A}}
\newcommand{\C}{\mathbb{C}}

\newcommand{\rd}[1]{{\color{red}#1}}

\title{Entanglement Cost of Erasure Correction in\\[0.05cm]Quantum MDS Codes}

\author{\IEEEauthorblockN{\,\\
{\large Kaushik Senthoor}
\\[0.05cm]QuTech, Delft University of Technology, Lorentzweg 1, 2628 CJ Delft, The Netherlands}
\thanks{Corresponding author: Kaushik Senthoor (r.k.senthoor@tudelft.nl)}}

\begin{document}

\maketitle

\begin{abstract}
In distributed quantum storage, physical qubits of a code will be stored across the network.
When qubits in one of the nodes are lost \textit{i.e.} when the node is erased, the remaining nodes need to communicate with a new node to replace the lost qubits.
Here, we look at the problem of how much entanglement cost is needed to perform such a distributed quantum erasure correction.
We focus on distributed quantum storage based on quantum maximum distance separable (MDS) codes.
We derive lower bounds on the entanglement cost when the quantum network used for the erasure correction has a star topology.
We show that the simple method of downloading the non-erased qudits and performing operations at a single node is optimal when the minimal number of non-erased nodes are accessed.
It remains to be seen what the entanglement cost will be when a non-minimal number of non-erased nodes are accessed.
The techniques used in this work can be developed further to study the entanglement cost of quantum erasure correction in more general code families and network topologies.
\end{abstract}

\section{Introduction}
\label{s:intro}
\thispagestyle{empty}

Over the last three decades, there has been growing research in the field of quantum information.
This research is mainly aimed at applications in cryptography, computation and metrology which can provide advantage over using classical information.
However, storing and processing quantum information within a single physical system is sometimes not practical.
Here we look at an interesting problem related to storing quantum information in a distributed system with efficient storage and communication costs.
Specifically, we study the entanglement cost (or the amount of quantum communication needed) to perform distributed erasure correction in quantum MDS codes.

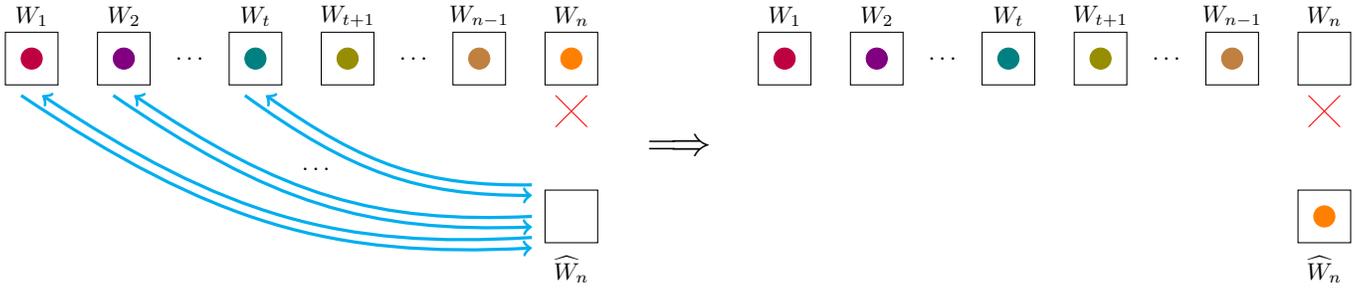
\begin{figure*}[t]
\begin{center}
\begin{tikzpicture}[xscale=0.7,yscale=1.4,every node/.style={scale=0.9}]
\draw (0,0) rectangle (1,0.5);
\draw [purple,fill=purple] (0.5,0.25) ellipse (0.2cm and 0.1cm);
\node at (0.5,0.65) {$W_1$};
\draw [<-,very thick,cyan] (10,-1.55) to [bend left=10] (0.3,-0.1);
\draw [->,very thick,cyan] (10,-1.45) to [bend left=10] (0.7,-0.1);
\draw (1.75,0) rectangle (2.75,0.5);
\draw [violet,fill=violet] (2.25,0.25) ellipse (0.2cm and 0.1cm);
\node at (2.25,0.65) {$W_2$};
\draw [<-,very thick,cyan] (10,-1.35) to [bend left=10] (2.05,-0.1);
\draw [->,very thick,cyan] (10,-1.25) to [bend left=10] (2.45,-0.1);
\node at (3.5,0.25) {$\hdots$};
\node at (5.9,-0.8) {$\hdots$};
\draw (4.25,0) rectangle (5.25,0.5);
\node at (4.75,0.65) {$W_t$};
\draw [teal,fill=teal] (4.75,0.25) ellipse (0.2cm and 0.1cm);
\draw [<-,very thick,cyan] (10,-1.05) to [bend left=10] (4.55,-0.1);
\draw [->,very thick,cyan] (10,-0.95) to [bend left=10] (4.95,-0.1);
\draw (6,0) rectangle (7,0.5);
\node at (6.5,0.65) {$W_{t+1}$};
\draw [olive,fill=olive] (6.5,0.25) ellipse (0.2cm and 0.1cm);
\node at (7.75,0.25) {$\hdots$};
\draw (8.5,0) rectangle (9.5,0.5);
\node at (9,0.65) {$W_{n-1}$};
\draw [brown,fill=brown] (9,0.25) ellipse (0.2cm and 0.1cm);
\draw (10.25,0) rectangle (11.25,0.5);
\node at (10.75,0.65) {$W_n$};
\draw [orange,fill=orange] (10.75,0.25) ellipse (0.2cm and 0.1cm);
\draw [red] (10.45,-0.4) -- (11.05,-0.1);
\draw [red] (10.45,-0.1) -- (11.05,-0.4);
\draw (10.25,-1.5) rectangle (11.25,-1);
\node at (10.75,-1.75) {$\widehat{W}_n$};
\node at (12.8,-0.6) {\LARGE$\Longrightarrow$};
\begin{scope}[xshift=14.3cm]
\draw (0,0) rectangle (1,0.5);
\draw [purple,fill=purple] (0.5,0.25) ellipse (0.2cm and 0.1cm);
\node at (0.5,0.65) {$W_1$};
\draw (1.75,0) rectangle (2.75,0.5);
\draw [violet,fill=violet] (2.25,0.25) ellipse (0.2cm and 0.1cm);
\node at (2.25,0.65) {$W_2$};
\node at (3.5,0.25) {$\hdots$};
\draw (4.25,0) rectangle (5.25,0.5);
\node at (4.75,0.65) {$W_t$};
\draw [teal,fill=teal] (4.75,0.25) ellipse (0.2cm and 0.1cm);
\draw (6,0) rectangle (7,0.5);
\node at (6.5,0.65) {$W_{t+1}$};
\draw [olive,fill=olive] (6.5,0.25) ellipse (0.2cm and 0.1cm);
\node at (7.75,0.25) {$\hdots$};
\draw (8.5,0) rectangle (9.5,0.5);
\node at (9,0.65) {$W_{n-1}$};
\draw [brown,fill=brown] (9,0.25) ellipse (0.2cm and 0.1cm);
\draw (10.25,0) rectangle (11.25,0.5);
\node at (10.75,0.65) {$W_n$};
\draw [red] (10.45,-0.4) -- (11.05,-0.1);
\draw [red] (10.45,-0.1) -- (11.05,-0.4);
\draw (10.25,-1.5) rectangle (11.25,-1);
\node at (10.75,-1.75) {$\widehat{W}_n$};
\draw [orange,fill=orange] (10.75,-1.25) ellipse (0.2cm and 0.1cm);
\end{scope}
\end{tikzpicture}
\captionsetup{justification=justified}
\caption{Download-and-return protocol for erasure correction of the $e$th physical qudit in an $[[n,2t-n]]_Q$ quantum MDS code with the $t$ helper nodes in $W_T$ using star network $H_1$ which has replacement node $\widehat{W}_e$ as hub.
Here, $e=n$ and $T=\{1,2,\hdots,t\}$.
The total number of qudits communicated equals $2t$.}
\label{fig:network_diag_1}
\end{center}
\end{figure*}

\subsection{Motivation}

Performing quantum error correction (QEC) in a distributed quantum storage will be an essential phenomenon in many applications of quantum computing and communication.
A clear example is modular quantum computing, which is an important alternative for individual quantum computers~\cite{grover97,cirac99}.
In recent times, this technique is being actively explored due to hardware challenges which place limits on the number of qubits a single node can hold.
Modular quantum computing architectures connect multiple nodes, each with a small number of qubits, through photonic interconnects.
Building hardware capabilities~\cite{aghaee25,weaver25,sunami25}, architectures~\cite{van10,monroe14,shapourian25} and distributed algorithms~\cite{main25,muralidharan25} needed for modular quantum computing are ongoing research topics.
Quantum data centers are another example where distributed QEC becomes important.
Today, classical data centers store large amounts of data, which can be accessed by end users on demand.
These data centers use distributed error correction for reliable and efficient data storage~\cite{dau18}.
In the future, we will need similar \textit{quantum} data centers to provide quantum applications as a service on a large scale.
For instance, Liu et al~\cite{liu23,liu24} consider quantum data centers for potential applications such as centralized generation of magic states and multi-party private quantum communication.

In distributed quantum storage, to handle errors in qubits stored within a single node, the qubits can be locally encoded using a quantum error correction code.
A small number of qubit errors in a node can be corrected through error correction cycles within that node.
However, a large number of correlated qubit errors may occur in a node due to disturbances such as cosmic ray events (CRE) or other environmental factors~\cite{cardani20,mcewen22,wu25}.
In such a case, to repair the node, we would need a quantum error correction code encoding across multiple nodes.
Xu et al~\cite{xu2022} have illustrated the idea of inter-node encoding to protect against cosmic ray events.
Their analysis showed that such inter-node encoding can effectively reduce the rate of CRE-induced erasures from once in 10 seconds to less than once in a month.
Scaling up the number of nodes, the amount of quantum communication in the network needed for the inter-node error-correction cycles will also increase.
Hence, it is important to know how to do inter-node quantum error correction with an efficient consumption of the EPR pairs generated in the network.

Distributed QEC is also relevant to quantum network applications using multipartite entanglement.
For cryptographic applications such as quantum secret sharing and conference key agreement, using multipartite entanglement is advantageous under certain settings~\cite{epping17,memmen23}.
Some qubits in a multipartite entangled state in a network can get lost before being used by the application, due to transmission loss or decoherence.
In this situation, the parties in the network may drop the multipartite state altogether and generate a new multipartite state.
As an alternative approach, the parties can simply replace the lost qubits using erasure correction of the existing multipartite state.
The entanglement cost needed to perform such an erasure correction is a key parameter in deciding between the two approaches.
Of particular interest are the absolutely maximally entangled (AME) states~\cite{helwig12}.
AME states is a special class of multipartite entangled states in which any bipartition of the parties is maximally entangled.
An AME state can also be defined as the logical state in a quantum MDS code encoding zero logical qudits.
For instance, AME states can be used in parallel teleportation protocols with more than two parties, in which the sender and the receiver can be decided after the distribution of the entangled state~\cite{helwig14}.

\subsection{Related work}

Any distributed QEC implementation needs to take into account inter-node connectivity and entanglement constraints.
Topological codes is a common class of quantum codes that is being considered for such implementations~\cite{nickerson13,bone24,singh24,chandra25}.
Sutcliffe et al~\cite{sutcliffe25} propose using hyperbolic Floquet codes, which have high rates and require only one EPR pair per syndrome measurement.
Nadkarni et al~\cite{nadkarni20} gives a construction for modified graph state codes in which the erasure correction operations are performed within a neighborhood of the erased node up to four edges away.
In quantum locally recoverable codes~\cite{golowich23,luo23}, an erased node accesses only a small number of nodes during erasure correction.
Delfosse et al~\cite{delfosse2024} introduced the quantum Gabidulin codes for inter-node quantum error correction in a stacked quantum memory with multi-qubit nodes.
We can also look at our problem in relation to regenerating codes~\cite{ramkumar22}, which give a reduced communication cost for erasure correction in distributed (classical) storage.

On the other hand, the entanglement cost needed for various distributed quantum operations (including distributed QEC) has also been studied.
These works look at the minimal entanglement needed for distributed quantum operations irrespective of the protocol used to perform the operations.
Yamasaki et al have derived bounds on the entanglement cost needed for quantum state construction~\cite{yamasaki17} and quantum state merging~\cite{yamasaki18} in exact and approximate settings.
Later, they extended these works to the entanglement cost of encoding and decoding quantum codes~\cite{yamasaki19}.
Devetak et al~\cite{devetak08} give information-theoretic expressions for minimal entanglement and quantum communication needed for redistribution of multipartite entanglement.
Hayden et al~\cite{hayden01} and related work~\cite{hayden03,yamasaki24} study the entanglement cost of quantum state transformations.

\subsection{Problem setting}

\begin{figure*}[t]
\begin{center}
\begin{tikzpicture}[xscale=0.7,yscale=1.4,every node/.style={scale=0.9}]
\draw (0,0) rectangle (1,0.5);
\draw [purple,fill=purple] (0.5,0.25) ellipse (0.2cm and 0.1cm);
\node at (0.5,0.65) {$W_1$};
\draw (1.75,0) rectangle (2.75,0.5);
\draw [violet,fill=violet] (2.25,0.25) ellipse (0.2cm and 0.1cm);
\node at (2.25,0.65) {$W_2$};
\draw [->,very thick,cyan] (1.2,-0.1) to [bend left=-30] (2.15,-0.1);
\draw [<-,very thick,cyan] (1,-0.1) to [bend left=-30] (2.35,-0.1);
\node at (3.5,0.25) {$\hdots$};
\node at (2.8,-0.35) {$\hdots$};
\draw (4.25,0) rectangle (5.25,0.5);
\node at (4.75,0.65) {$W_{t-1}$};
\draw [magenta,fill=magenta] (4.75,0.25) ellipse (0.2cm and 0.1cm);
\draw [->,very thick,cyan] (0.6,-0.1) to [bend left=-25] (4.65,-0.1);
\draw [<-,very thick,cyan] (0.4,-0.1) to [bend left=-25] (4.85,-0.1);
\draw (6,0) rectangle (7,0.5);
\node at (6.5,0.65) {$W_t$};
\draw [->,very thick,cyan] (0.2,-0.1) to [bend left=-25] (6.4,-0.1);
\draw [<-,very thick,cyan] (0,-0.1) to [bend left=-25] (6.6,-0.1);
\draw [teal,fill=teal] (6.5,0.25) ellipse (0.2cm and 0.1cm);
\node at (7.75,0.25) {$\hdots$};
\draw (8.5,0) rectangle (9.5,0.5);
\node at (9,0.65) {$W_{n-1}$};
\draw [brown,fill=brown] (9,0.25) ellipse (0.2cm and 0.1cm);
\draw (10.25,0) rectangle (11.25,0.5);
\node at (10.75,0.65) {$W_n$};
\draw [orange,fill=orange] (10.75,0.25) ellipse (0.2cm and 0.1cm);
\draw [red] (10.45,-0.4) -- (11.05,-0.1);
\draw [red] (10.45,-0.1) -- (11.05,-0.4);
\draw (10.25,-1.5) rectangle (11.25,-1);
\node at (10.75,-1.75) {$\widehat{W}_n$};
\draw [->,very thick,cyan] (-0.2,-0.1) to [bend left=-18] (10,-1.25);
\node at (12.8,-0.6) {\LARGE$\Longrightarrow$};
\begin{scope}[xshift=14.3cm]
\draw (0,0) rectangle (1,0.5);
\draw [purple,fill=purple] (0.5,0.25) ellipse (0.2cm and 0.1cm);
\node at (0.5,0.65) {$W_1$};
\draw (1.75,0) rectangle (2.75,0.5);
\draw [violet,fill=violet] (2.25,0.25) ellipse (0.2cm and 0.1cm);
\node at (2.25,0.65) {$W_2$};
\node at (3.5,0.25) {$\hdots$};
\draw (4.25,0) rectangle (5.25,0.5);
\node at (4.75,0.65) {$W_{t-1}$};
\draw [magenta,fill=magenta] (4.75,0.25) ellipse (0.2cm and 0.1cm);
\draw (6,0) rectangle (7,0.5);
\node at (6.5,0.65) {$W_t$};
\draw [teal,fill=teal] (6.5,0.25) ellipse (0.2cm and 0.1cm);
\node at (7.75,0.25) {$\hdots$};
\draw (8.5,0) rectangle (9.5,0.5);
\node at (9,0.65) {$W_{n-1}$};
\draw [brown,fill=brown] (9,0.25) ellipse (0.2cm and 0.1cm);
\draw (10.25,0) rectangle (11.25,0.5);
\node at (10.75,0.65) {$W_n$};
\draw [red] (10.45,-0.4) -- (11.05,-0.1);
\draw [red] (10.45,-0.1) -- (11.05,-0.4);
\draw (10.25,-1.5) rectangle (11.25,-1);
\node at (10.75,-1.75) {$\widehat{W}_n$};
\draw [orange,fill=orange] (10.75,-1.25) ellipse (0.2cm and 0.1cm);
\end{scope}
\end{tikzpicture}
\captionsetup{justification=justified}
\caption{Download-and-return protocol for erasure correction of the $e$th physical qudit in an $[[n,2t-n]]_Q$ quantum MDS code with the $t$ helper nodes in $W_T$ using star network $H_2$ which has one of the helper nodes $W_{j_1}$ ($j_1\in T$) as hub.
Here $e=n$, $T=\{1,2,\hdots,t\}$ and $j_1=1$.
The total number of qudits communicated equals $2t-1$.}
\label{fig:network_diag_2}
\end{center}
\end{figure*}
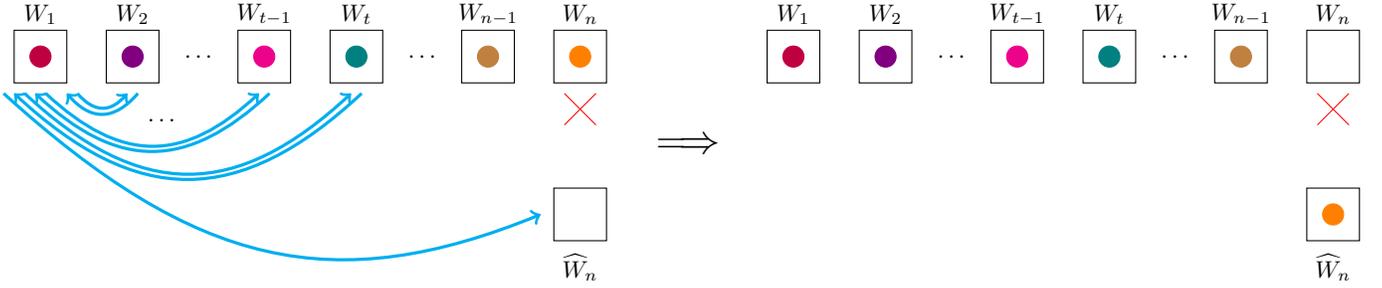

In this paper, we consider only the distributed erasure correction in quantum MDS codes.
This is due to two reasons.
The information-theoretic structure of quantum MDS codes is already well studied in the literature~\cite{huber20,grassl22}.
Secondly, the quantum MDS codes give the lowest possible storage overhead in distributed quantum storage (for suitable dimensions of the quantum state stored).
Consider a distributed quantum storage system with $n$ nodes which needs to be resilient against erasure of any $n-t$ nodes, where $1<t<n\leq 2t$.
Let $F$ be the dimension of the quantum state encoded and stored in this system.
In the case when $F=Q^{2t-n}$, where $Q$ is some prime power, an $[[n,2t-n]]_Q$ quantum MDS code encoding the state and storing one ($Q$-dimensional) physical qudit in each node gives the lowest storage overhead.
See Section~\ref{ss:qmds} for the definition of quantum MDS codes and some of their properties.

The new node which replaces an erased node in a distributed quantum storage is called the \textit{replacement node}.
At the end of erasure correction, the $n-1$ non-erased nodes and the replacement node together should contain the original encoded state.
We observe that quantum communication between the replacement node and any $t$ non-erased nodes is necessary and sufficient for erasure correction in an $[[n,2t-n]]_Q$ quantum MDS code (see Lemma~\ref{lm:d_geq_t} and \ref{lm:t_is_suff}).
As the main result of this paper, we find the entanglement cost for erasure correction happening over a quantum network of star topology containing any $t$ non-erased nodes and the replacement node.
The non-erased nodes in this quantum network are called \textit{helper nodes}.
We define the entanglement cost as the minimum number of $Q$-dimensional qudits which need to be communicated in the network for successful erasure correction.
This definition of entanglement cost is similar to the \textit{total graph-associated entanglement cost} defined in \cite[Section~II-B]{yamasaki17}.

\subsection{Contributions}

We obtain the entanglement cost of erasure correction in an $[[n,2t-n]]_Q$ quantum MDS code as follows.
The quantum network (with star topology) used for the erasure correction contains the replacement node and any $t$ helper nodes.
Any classical communication among the replacement node and the $n-1$ non-erased nodes is allowed.
\begin{itemize}
\item If the replacement node is the hub in the quantum network, the entanglement cost is $2t$ qudits.
\item If one of the $t$ helper nodes is the hub in the quantum network, the entanglement cost is $2t-1$ qudits.
\end{itemize}
First, we derive the lower bound on entanglement cost for each of the two star topologies.
When the replacement node is the hub, the lower bound can be simply achieved by transferring qudits from the $t$ helper nodes to the hub, performing operations locally and sending back the $t$ qudits.
Similarly, when one of the helper nodes is the hub, the lower bound can be achieved by transferring qudits from the $t-1$ other non-erased nodes to the hub, performing operations locally, sending back the downloaded $t-1$ qudits and sending the recovered qudit to the replacement node.
We refer to such erasure correction protocols as \textit{download-and-return} protocols (see Fig.~\ref{fig:network_diag_1} and Fig.~\ref{fig:network_diag_2}).

Note that our results also apply to the AME states, by taking the specific case of $n=2t$.
The techniques used in this paper should also be helpful in studying the entanglement cost of erasure correction for more general network topologies and code families.

\subsection{Approach}

In the work by Yamasaki et al~\cite[Theorem~7]{yamasaki19}, a lower bound on the entanglement cost of encoding of a quantum code over a tree network is derived.
They model the distribution of the encoded physical qudits over a network as an LOCC protocol using a resource state having maximally entangled states.
Then, by applying the monotonicity of Schmidt rank under LOCC, they give a lower bound on the number of EPR pairs needed.
They also show that the lower bound on the entanglement cost is achievable using a specific algorithm.

Here, we follow a similar approach for deriving the lower bound on the entanglement cost of quantum erasure correction.
Quantum erasure correction is formulated as a circuit with an LOCC protocol acting on a combination of the encoded state and the resource state.
Then, by applying the monotonicity property of Schmidt rank under LOCC, we obtain the lower bounds for two different star networks.
There is a key difference in our approach for deriving the lower bound compared to that of Yamasaki et al.
We assume that a suitable unitary operator is applied and then inverted on the nodes not directly involved in the LOCC protocol.
Such an assumption gives us a better bound than applying the Schmidt rank monotonicity with just the LOCC protocol.

\subsection{Organization}
In Section~\ref{s:bg}, we give the necessary background on LOCC protocols and quantum MDS codes.
We also state and prove some properties needed to build our model.
In Section~\ref{s:eras_corr}, we explain our problem setting in detail.
Here, we model the quantum erasure correction as a circuit using an LOCC protocol.
In Section~\ref{s:ent_cost}, we derive lower bounds on the entanglement cost needed in two different star network topologies.
We also show that these bounds are optimal.
In Section~\ref{s:conc}, we summarize the paper and discuss the future directions of our work.

\section{Background}
\label{s:bg}

\subsection{Notation}

A finite-dimensional Hilbert space of dimension $M$ is denoted as $\Hs_M$.
We denote the set of real numbers as $\mathbb{R}$ and the set of complex numbers as $\C$.
For any positive integer $m$, we write $[m]:=\{1,2,\hdots,m\}$ and $\A_m=\{0,1,\hdots,m-1\}$.
In this paper, we mostly use $[m]$ to index elements 1 to $m$, and $\A_m$ to denote the set of computational basis states in $\Hs_m$.
For any two positive integers $m_1$ and $m_2$ such that $m_1\leq m_2$, we write $[m_1,m_2]=\{m_1,m_1+1,\hdots,m_2\}$.

For a vector $\ul{x}=(x_1,x_2,\hdots,x_\ell)\in\A_Q^\ell$, the computational basis state $\ket{x_1}\ket{x_2}\hdots\ket{x_\ell}\in\Hs_Q^\ell$ is also denoted as $\ket{\ul{x}}$.
Let $P_1$, $P_2$, $\hdots$~, $P_m$ be $m$ different parties.
For any subset $J=\{j_1,j_2,\hdots,j_\ell\}\subseteq[m]$, the set of parties $\{P_{j_1}$, $P_{j_2}$, $\hdots$~, $P_{j_\ell}\}$ is denoted as $P_J$.
For any party $P$, we denote the identity operator acting on the subsystem held by $P$ as $\idmtx_P$.

Let $\ket{\phi}_{AB}\in\Hs_{M_A}\otimes\Hs_{M_B}$ be a pure state jointly held by two parties $A$ and $B$.
Here $M_A$ (or $M_B$) is the dimension of the subsystem held by $A$ (or $B$). 
The state of the subsystem held by party $A$ is given by the density matrix
\begin{equation}
\rho(A)_\phi:=\text{Tr}_B\ket{\phi}\bra{\phi}_{\,AB}.
\end{equation}
Let the Schmidt decomposition of $\ket{\phi}$ between $A$ and $B$ be
\begin{equation}
\ket{\phi}_{AB}=\sum_{i=1}^{N}\lambda_i\ket{a_i}_A\ket{b_i}_B.
\label{eq:schmidt_decomp}
\end{equation}
where $N=\min\{M_A,M_B\}$, $\lambda_i\geq 0$ for all $i\in[N]$ and $\{\ket{a_i}\}_{i=1}^N$ (or $\{\ket{b_i}\}_{i=1}^N$) gives a set of orthonormal pure states in $\Hs_{M_A}$ (or $\Hs_{M_B}$).
For the state $\ket{\phi}$ with Schmidt decomposition in Eq.~\eqref{eq:schmidt_decomp}, the subsystem held by $A$ is given by the density operator
\begin{eqnarray}
\rho(A)_\phi=\sum_{i=1}^N\lambda_i^2\ket{a_i}\bra{a_i}.
\end{eqnarray}
We denote the vector of Schmidt coefficients as \begin{eqnarray}
\ul{\lambda}(A)_\phi:=(\lambda_1,\lambda_2,\hdots,\lambda_N).
\end{eqnarray}
Note that $\ul{\lambda}(A)_\phi=\ul{\lambda}(B)_\phi$ and the entries in $\ul{\lambda}(A)_\phi$ satisfy
\begin{equation}
\sum_{i=1}^{N}\lambda_i^2=1.
\end{equation}

The Schmidt rank of the state $\ket{\phi}_{AB}$ for the subsystem held by $A$ is defined as the total number of positive entries in $\ul{\lambda}(A)_\phi$.
It is denoted by $\sr(A)_\phi$.
The von Neumann entropy of $\ket{\phi}$ for the subsystem held by $A$ is defined as
\begin{eqnarray}
\mathbf{S}(A)_\phi:=\sum_{i=1}^{N}\lambda_i^2\,\log\frac{1}{\lambda_i^2}\,.
\end{eqnarray}

\subsection{LOCC protocols}

We refer to a protocol of local operations and classical communication involving $m$ parties $P_1$, $P_2$, $\hdots$, $P_m$ as an LOCC$(P_1,P_2,\hdots,P_m)$ protocol or simply as an LOCC protocol.
In this article, we consider only deterministic LOCC protocols and not stochastic LOCC protocols.
In other words, if an LOCC protocol $L$ is said to transform the state $\ket{\psi_\text{in}}$ to $\ket{\psi_\text{out}}$, the probability of obtaining the state $\ket{\psi_\text{out}}$ from the input state $\ket{\psi_\text{in}}$ is 1.

We assume that each party in an LOCC protocol has an unlimited number of ancilla qudits.
This means any local CPTP operator can be seen as a unitary operator due to Stinespring dilation~\cite[Theorem~21.2]{bertlmann23}.
Similarly, any local measurement can be seen as a projective measurement~\cite[Section~2.2.8]{nielsen10}.

Quantum teleportation is one of the well-known two-party LOCC protocols.
This protocol shows how to use an entangled state as a resource state to perform quantum communication.
Bennett et al~\cite{bennett93} have shown that an arbitrary quantum state of dimension $M$ can be communicated (in either direction) between two parties $A$ and $B$ with LOCC if they already share the maximally entangled state
\begin{equation}
\ket{M^+}_{A,B}:=\frac{1}{\sqrt{M}}\,\sum_{i=0}^{M-1}\ket{i}_A\ket{i}_B
\,\,\in\,\,\Hs_M\otimes\Hs_M.
\end{equation}
This idea can also be extended to communicating a sequence of quantum states between $A$ and $B$ as follows.
\begin{lemma}[Quantum teleportation]
\label{lm:q_tele}
A sequence of arbitrary quantum states of dimensions ${M_1,M_2,\hdots,M_\ell}$ can be communicated (in either direction) between parties $A$ and $B$ using an LOCC protocol if the two parties already share the maximally entangled state $\ket{M^+}_{A,B}$ where $M\geq M_1 M_2\hdots M_\ell$.
\end{lemma}

The following lemma shows the monotonicity of Schmidt rank under LOCC.
Schmidt rank of a party either decreases or remains the same after an LOCC protocol.
\begin{lemma}\cite[Lemma~1]{lo01}
Let $\ket{\psi_\text{in}}_{P_1P_2\hdots P_m}$ and $\ket{\psi_\text{out}}_{P_1P_2\hdots P_m}$ be two $m$-partite pure states.
If an LOCC$(P_1,P_2,\hdots P_m)$ protocol $L$ transforms the state $\ket{\psi_\text{in}}$ to $\ket{\psi_\text{out}}$, then
\begin{equation}
\sr(P_J)_{\psi_\text{out}}\leq\,
\sr(P_J)_{\psi_\text{in}}
\end{equation}
for any $J\subseteq[m]$.
\label{lm:monotone_sr}
\end{lemma}

The following lemma says that an $m$-partite LOCC protocol can be replaced with an $(m-1)$-partite LOCC protocol if the Schmidt coefficients of the $m$th party in the input and output states remain the same and are all equal. 
\begin{lemma}
\label{lm:locc_ent_unchanged}
Let $\ket{\psi_\text{in}}_{P_1P_2\hdots P_m}$ and $\ket{\psi_\text{out}}_{P_1P_2\hdots P_m}$ be two $m$-partite pure states such that
\begin{equation}
\ul{\lambda}(P_m)_{\psi_\text{out}}=\ul{\lambda}(P_m)_{\psi_\text{in}}
=\left(\frac{1}{\sqrt{Q}},\frac{1}{\sqrt{Q}},\hdots,\frac{1}{\sqrt{Q}}\right)\in\mathbb{R}^Q.
\end{equation}
An LOCC$(P_1,P_2,\hdots,P_m)$ protocol $L_1$ which transfroms $\ket{\psi_\text{in}}$ to $\ket{\psi_\text{out}}$ exists if and only if an LOCC$(P_1,P_2,\hdots,P_{m-1})$ protocol $L_2$ exists such that $\ket{\psi_\text{out}}$ can be obtained by
\begin{enumerate}
\item applying $(L_2)_{P_1P_2\hdots P_{m-1}}\,\otimes\,\idmtx_{P_m}$ on $\ket{\psi_\text{in}}$ to obtain the state $\ket{\psi_\text{mid}}$ and
\item applying $\idmtx_{P_1P_2\hdots P_{m-1}}\,\otimes\,G(\ul{v})_{P_m}$ on $\ket{\psi_\text{mid}}$.
Here $G(\ul{v})$ is a unitary operator dependent on the local measurements $\ul{v}$ obtained from the execution of $L_2$ in the previous step.
\end{enumerate}
\end{lemma}
\begin{proof}
See Appendix~\ref{ap:locc}.
\end{proof}

\subsection{Quantum MDS codes}
\label{ss:qmds}

An $((n,K))_Q$ quantum code is a subspace of $\Hs_Q^{\otimes n}$ with dimension $K$ where $1\leq K\leq Q^n$.
Each state in the code is defined over $n$ physical qudits, each of dimension $Q$.
When $K=Q^k$, we can also denote this code as an $[[n,k]]_Q$ quantum code.
We refer to such a code as encoding $k$ logical qudits into $n$ physical qudits.

An encoding operator for an $[[n,k]]_Q$ quantum code can be defined as a unitary operator $E$ which gives an encoding map
\begin{equation}
\ket{\ul{s}}\mapsto E\left(\ket{\ul{s}}\ket{0}^{\otimes(n-k)}\right)
\label{eq:enc_basis_state}
\end{equation}
for all $\ul{s}\in\A_Q^k$.
Here $\ket{\ul{s}}$ is a computational basis state of $\Hs_Q^{\otimes k}$ and $\ket{0}$ is an ancilla qudit of dimension $Q$.
Note that the same code can have different encoding operators depending on the mapping chosen.

\begin{definition}[Distance]\cite[Section III-A]{scott04}
Consider an $((n,K))_Q$ quantum code with $K>1$ having an orthonormal basis of code states $\{\ket{\phi_i}\}_{\ell=1}^{K}$.
The code is said to have distance $D$ if for any operator $\Gamma$ acting on $D-1$ or less physical qudits,
\begin{equation}
\bra{\phi_i}\Gamma\ket{\phi_j}=c(\Gamma)\,\delta_{i,j}
\label{eq:k_great_1_distance_cond}
\end{equation}
where $c(\Gamma)$ is a function independent of $i$ and $j$.
Here $\delta_{i,j}$ is the function whose value is 1 if $i=j$ and 0 otherwise.

An $((n,1))_Q$ quantum code with the code state $\ket{\phi}$ is said to have distance $D$ if for any operator $\Gamma$ acting on $D-1$ or less physical qudits,
\begin{equation}
\bra{\phi}\Gamma\ket{\phi}=Q^{-n}\Tr(\Gamma).
\label{eq:k_equals_1_distance_cond}
\end{equation}
Here $\Tr(\Gamma)$ indicates the trace of the operator $\Gamma$.
\end{definition}

Note that the definition of distance for quantum codes with $K>1$ is different from those with $K=1$.
For $K=1$, the condition in Eq.~\eqref{eq:k_great_1_distance_cond} becomes trivially true.
For a quantum code with distance $D$, erasure of any $D-1$ or fewer physical qudits can be corrected.

From the quantum Singleton bound~\cite[Theorem~2]{rains99}, we know that
\begin{equation}
2(D-1)\leq n-\log_Q K.
\label{eq:q_singleton_bound}
\end{equation}
A quantum code having the maximum possible distance given by this bound is called a quantum maximum distance separable code (also called as quantum MDS code or QMDS code).
\begin{definition}[QMDS code]
\cite[Section~III]{rains99}
An $[[n,k]]_Q$ quantum code with its distance $D$ satisfying
\begin{equation}
2(D-1)=n-k
\label{eq:q_singleton_bound_equal}
\end{equation}
is called a quantum MDS code (or a QMDS code).
\end{definition}
The $[[5,1,3]]_2$ code given by \cite{laflamme96,bennett96} is a well-known example of quantum MDS codes.
Note that $[[n,2t-n]]_Q$ quantum MDS code for given $n$ and $t$ does not always exist for qubits \textit{i.e.} with $Q=2$.
Hence, we consider codes using the more general $Q$-dimensional qudits where $Q\geq 2$.
When $Q$ is a power of 2 such that $Q=2^w$, then each node can be considered to be storing $w$ qubits.
For more details on quantum MDS codes, their properties and constructions, see~\cite{rains99,grassl04,grassl15,huber20,alsina21}.

Let $t$ be an integer such that $1<t<n\leq 2t$.
Consider an $[[n,2t-n]]_Q$ quantum MDS code.
This code has distance
\begin{equation}
D=n-t+1
\end{equation}
and can correct the erasure of any $n-t$ or fewer qudits.
Here we assume that the $n$ physical qudits of the $[[n,2t-n]]_Q$ quantum MDS code are stored in $n$ different nodes.
The node containing the $i$th physical qudit is denoted as $W_i$.
We denote the set of nodes given by a subset $J\subseteq[n]$ as $W_J$.

It has been shown that any $n-t$ physical qudits in the encoded state of the quantum MDS code are in the maximally mixed state~\cite{huber20,grassl22}.
The following lemma uses this fact to provide an encoding circuit for a given encoding operator $E$ of the quantum MDS code.

\begin{lemma}
\label{lm:unit_qmds}
Consider an $[[n,2t-n]]_Q$ quantum MDS code with an encoding operator $E$.
For any $T\subseteq[n]$ of size $|T|=t$, there exists a unitary operator $U_T^{(E)}$ over $\Hs_Q^{\otimes t}$ such that the encoding is given by
\begin{flalign}
&\,\,E\left(\ket{\ul{s}}\ket{0}^{\otimes(2n-2t)}\right)&\nonumber
\\&\,\,\,\,\,\,=\frac{1}{\sqrt{Q^{n-t}}}\sum_{\substack{r_1,r_2,\hdots,\\r_{n-t}\,\in\,\A_Q}}
\!\!U_T^{(E)}\!\left(\ket{\ul{s}}\ket{r_1}\!\ket{r_2}\hdots\ket{r_{n-t}}\right)_{W_T}&
\nonumber
\\[-0.5cm]&\hspace{4.2cm}\big(\ket{r_1}\!\ket{r_2}\hdots\ket{r_{n-t}}\big)_{W_{[n]\setminus{T}}}\!\!\!\!\!\!\!&
\label{eq:unit_qmds}
\end{flalign}
for all $\ul{s}\in\A_Q^{2t-n}$.
\end{lemma}
\begin{proof}
See Appendix~\ref{ap:qmds}.
\end{proof}

We can also use the operator $U_T^{(E)}$ defined above to construct smaller quantum MDS codes.
\begin{lemma}
\label{lm:unit_p1_qmds}
Let $U_T^{(E)}$ be the unitary operator as defined in Lemma~\ref{lm:unit_qmds}.
The $[[t+1,t-1]]_Q$ quantum code with the encoding
\begin{flalign}
&\hspace{0.2cm}\ket{s_1}\ket{s_2}\hdots\ket{s_{t-1}}&\nonumber
\\*[0.05cm]&\hspace{1cm}\mapsto\,\frac{1}{\sqrt{Q}}
\sum_{r\in\A_Q}U_T^{(E)}\big(\!\ket{s_1}\ket{s_2}\hdots\ket{s_{t-1}}\ket{r}\!\big)
\,\ket{r}&\nonumber
\\*[-0.2cm]&\hspace{4.1cm}\forall (s_1,s_2,\hdots,s_{t-1})\in\A_Q^{t-1}\!\!\!\!\!&
\label{eq:ud_p1_qmds_enc}
\end{flalign}
is a quantum MDS code.
\end{lemma}
\begin{proof}
See Appendix~\ref{ap:qmds}.
\end{proof}
The above two lemmas will be useful in deriving the lower bounds on entanglement cost later in this paper.

\section{Erasure correction in QMDS codes}
\label{s:eras_corr}

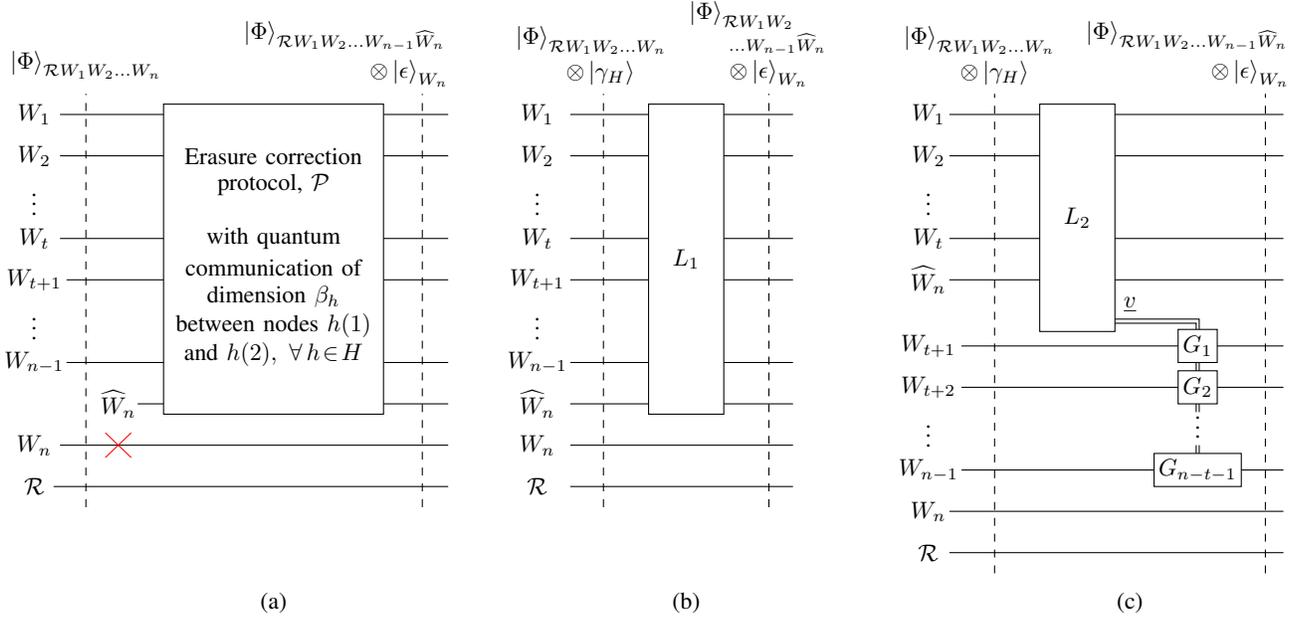
\begin{figure*}[t]
\begin{center}
\begin{tikzpicture}[xscale=0.4,yscale=0.55,every node/.style={scale=0.9}]
\begin{scope}[xshift=-20cm,xscale=4.3]
\node at (2.45,3.2) {(a)};
\node at (0.6,15) {$W_1$};
\draw (0.8,15) -- (1.6,15);
\node at (0.6,14) {$W_2$};
\draw (0.8,14) -- (1.6,14);
\node at (0.6,13) {$\vdots$};
\node at (0.6,12) {$W_t$};
\draw (0.8,12) -- (1.6,12);
\node at (0.6,11) {$W_{t+1}$};
\draw (0.85,11) -- (1.6,11);
\node at (0.6,10) {$\vdots$};
\node at (0.6,9) {$W_{n-1}$};
\draw (0.85,9) -- (1.6,9);
\node at (0.6,7) {$W_n$};
\draw (0.8,7) -- (3.8,7);
\node at (0.6,6) {$\mathcal{R}$};
\draw (0.75,6) -- (3.8,6);
\node at (1,16.2) {$\ket{\Phi}_{\mathcal{R}W_1 W_2\hdots W_n}$};
\draw [dashed] (1,5.5) -- (1,15.5);
\rd{\draw (1.15,6.7) -- (1.35,7.3);
\draw (1.15,7.3) -- (1.35,6.7);}
\node at (1.25,8) {$\widehat{W}_n$};
\draw (1.4,8) -- (1.6,8);
\draw (1.6,7.75) rectangle (3.3,15.25);
\node at (2.45,14) {Erasure correction};
\node at (2.45,13.3) {protocol, $\mathcal{P}$};
\node at (2.45,12) {with quantum};
\node at (2.45,11.3) {communication of};
\node at (2.45,10.6) {dimension $\beta_h$};
\node at (2.45,9.9) {between nodes $h(1)$};
\node at (2.45,9.2) {and $h(2),\,\,\forall\,h\!\in\!H$};
\draw (3.3,15) -- (3.8,15);
\draw (3.3,14) -- (3.8,14);
\draw (3.3,12) -- (3.8,12);
\draw (3.3,11) -- (3.8,11);
\draw (3.3,9) -- (3.8,9);
\draw (3.3,8) -- (3.8,8);
\draw [dashed] (3.6,5.5) -- (3.6,15.5);
\node at (3,16.9) {$\ket{\Phi}_{\mathcal{R}W_1 W_2\hdots W_{n-1}\widehat{W}_n}$};
\node at (3.5,16) {$\otimes\ket{\epsilon}_{W_n}$};
\end{scope}
\node at (4.25,3.2) {(b)};
\node at (-0.7,15) {$W_1$};
\draw (0.4,15) -- (3,15);
\node at (-0.7,14) {$W_2$};
\draw (0.4,14) -- (3,14);
\node at (-0.7,13) {$\vdots$};
\node at (-0.7,12) {$W_t$};
\draw (0.4,12) -- (3,12);
\node at (-0.7,11) {$W_{t+1}$};
\draw (0.4,11) -- (3,11);
\node at (-0.7,10) {$\vdots$};
\node at (-0.7,9) {$W_{n-1}$};
\draw (0.4,9) -- (3,9);
\node at (-0.7,7) {$W_n$};
\draw (0.4,7) -- (7.8,7);
\node at (-0.7,6) {$\mathcal{R}$};
\draw (0.4,6) -- (7.8,6);
\node at (1.1,16.8) {$\ket{\Phi}_{\mathcal{R}W_1 W_2\hdots W_n}$};
\node at (1.4,16) {$\otimes\ket{\gamma_H}$};
\node at (-0.7,8) {$\widehat{W}_n$};
\draw (0.4,8) -- (3,8);
\draw [dashed] (1.5,5.5) -- (1.5,15.5);
\draw (3,7.75) rectangle (5.5,15.25);
\node at (4.25,11.5) {$L_1$};
\draw (5.5,15) -- (7.8,15);
\draw (5.5,14) -- (7.8,14);
\draw (5.5,12) -- (7.8,12);
\draw (5.5,11) -- (7.8,11);
\draw (5.5,9) -- (7.8,9);
\draw (5.5,8) -- (7.8,8);
\draw [dashed] (7,5.5) -- (7,15.5);
\node at (6.1,17.4) {$\ket{\Phi}_{\mathcal{R}W_1 W_2}$};
\node at (7.3,16.7) {$_{\hdots W_{n-1}\widehat{W}_n}$};
\node at (7,16) {$\otimes\ket{\epsilon}_{W_n}$};
\begin{scope}[xshift=13cm]
\node at (6,3.2) {(c)};
\node at (6,3.1) {};
\node at (-0.7,15) {$W_1$};
\draw (0,15) -- (3,15);
\node at (-0.7,14) {$W_2$};
\draw (0,14) -- (3,14);
\node at (-0.7,13) {$\vdots$};
\node at (-0.7,12) {$W_t$};
\draw (0,12) -- (3,12);
\node at (-0.7,11) {$\widehat{W}_n$};
\draw (0,11) -- (3,11);
\node at (-0.7,9.4) {$W_{t+1}$};
\draw (0.4,9.4) -- (7.6,9.4);
\node at (-0.7,8.4) {$W_{t+2}$};
\draw (0.4,8.4) -- (7.6,8.4);
\node at (-0.7,7.4) {$\vdots$};
\node at (-0.7,6.4) {$W_{n-1}$};
\draw (0.4,6.4) -- (6.8,6.4);
\node at (-0.7,5.4) {$W_n$};
\draw (0,5.4) -- (11.1,5.4);
\node at (-0.7,4.4) {$\mathcal{R}$};
\draw (0,4.4) -- (11.1,4.4);
\node at (1,16.8) {$\ket{\Phi}_{\mathcal{R}W_1 W_2\hdots W_n}$};
\node at (1.5,16) {$\otimes\ket{\gamma_H}$};
\draw [dashed] (1.5,3.9) -- (1.5,15.5);
\draw (3,9.75) rectangle (5.5,15.25);
\node at (4.25,12.5) {$L_2$};
\draw (5.5,15) -- (11.1,15);
\draw (5.5,14) -- (11.1,14);
\draw (5.5,12) -- (11.1,12);
\draw (5.5,11) -- (11.1,11);
\draw (5.5,9.95) -- (8.2,9.95) -- (8.2,9.8);
\draw (5.5,10.05) -- (8.3,10.05) -- (8.3,9.8);
\node at (6,10.4) {$\ul{v}$};
\draw (7.6,9.8) rectangle (8.9,9);
\node at (8.25,9.4) {$G_1$};
\draw (8.2,9) -- (8.2,8.8);
\draw (8.3,9) -- (8.3,8.8);
\draw (8.9,9.4) -- (11.1,9.4);
\draw (7.6,8.8) rectangle (8.9,8);
\node at (8.25,8.4) {$G_2$};
\draw (8.9,8.4) -- (11.1,8.4);
\draw (8.2,8) -- (8.2,7.8);
\draw (8.3,8) -- (8.3,7.8);
\node at (8.25,7.6) {$\vdots$};
\draw (8.2,7) -- (8.2,6.8);
\draw (8.3,7) -- (8.3,6.8);
\draw (6.8,6.8) rectangle (9.7,6);
\node at (8.25,6.4) {$G_{n-t-1}$};
\draw (9.7,6.4) -- (11.1,6.4);
\draw [dashed] (10.5,3.9) -- (10.5,15.5);
\node at (7.9,16.9) {$\ket{\Phi}_{\mathcal{R}W_1 W_2\hdots W_{n-1}\widehat{W}_n}$};
\node at (10,16) {$\otimes\ket{\epsilon}_{W_n}$};
\end{scope}
\end{tikzpicture}
\captionsetup{justification=justified}
\caption{(a) Erasure correction to replace the erased node $W_n$ with $\widehat{W}_n$.
(b) Erasure correction using LOCC$(W_1$, $W_2$, $\hdots$\,,\,$W_n$,$\widehat{W}_n)$ protocol $L_1$ with resource state $\ket{\gamma_H}$.
(c) Erasure correction using LOCC$(W_1$, $W_2$, $\hdots$\,,\,$W_t$, $\widehat{W}_n)$ protocol $L_2$ and the unitary operators $G_1$, $G_2$, $\hdots$\,,\,$G_{n-t-1}$ with resource state $\ket{\gamma_H}$.}
\label{fig:eras_corr_as_locc}
\end{center}
\end{figure*}

In this section, we first model the erasure correction as a circuit involving an LOCC protocol (see Theorem~\ref{th:ec_locc}).
Then, we look at the Schmidt ranks of a specific subsystem for the states just before and after the LOCC protocol (see Lemmas~\ref{lm:sr_psi_in} and \ref{lm:sr_psi_out}).
In the next section, we will use these Schmidt ranks to derive a lower bound on the entanglement cost of erasure correction.

Here we consider a distributed quantum storage with $n$ nodes $W_1$, $W_2$, $\hdots$, $W_n$.
Each physical qudit in the logical state of an $[[n,2t-n]]_Q$ QMDS code is stored in one of the $n$ nodes.

Assume that node $W_e$ (where $e\in[n]$) is lost \textit{i.e.} the $e$th physical qudit is erased.
We can perform erasure correction to replace the erased qudit.
Using erasure correction, some of the $n-1$ non-erased nodes can perform joint quantum operations on their qudits along with an ancilla qudit and generate a new replacement qudit.
We refer to the non-erased nodes involved in these quantum operations as helper nodes.
Taking the replacement qudit in place of the $e$th qudit should yield the logical state.
This replacement qudit is stored in a replacement node $\widehat{W}_e$ in place of the lost node.
For any logical state $\ket{\phi}$ of the quantum MDS code, the erasure correction process should take the input state
\begin{equation}
\ket{\phi}_{W_1 W_2\hdots W_n}\ket{0}_{\widehat{W}_e}
\end{equation}
to the output state
\begin{equation}
\ket{\phi}_{W_1 W_2\hdots W_{e-1}\widehat{W}_e W_{e+1}\hdots W_n}\ket{\epsilon}_{W_e}.
\end{equation}
Here $\ket{\epsilon}$ is some pure state in $\Hs_Q$.

The following two lemmas talk about the number of helper nodes needed for erasure correction.

\begin{lemma}
\label{lm:d_geq_t}
Consider a distributed quantum storage with $n$ nodes using an $[[n,2t-n]]_Q$ quantum MDS code.
To replace a lost node, the replacement node needs quantum communication with at least $t$ of the other $n-1$ nodes.
\end{lemma}
\begin{proof}
See Appendix~\ref{ap:qmds}.
\end{proof}

\begin{lemma}
\label{lm:t_is_suff}
Consider a distributed quantum storage with $n$ nodes using an $[[n,2t-n]]_Q$ quantum MDS code.
To replace a lost node, it is sufficient for the replacement node to communicate with any $t$ of the other $n-1$ nodes.
\end{lemma}
\begin{proof}
See Appendix~\ref{ap:qmds}.
\end{proof}

The above two lemmas show that we need $d\geq t$ helper nodes for successful erasure correction.
In this paper, we study the entanglement cost of erasure correction when $d=t$.
The entanglement cost when $d>t$ will be studied as part of future work.

Let $T=\{j_1,j_2,\hdots,j_t\}\subseteq[n]\setminus\{e\}$ of size $|T|=t$ indicate the set of helper nodes.
Then the quantum network used for the erasure correction is given by a connected graph $(V,H)$ with the vertices given by
\begin{equation}
V=\{\widehat{W}_e,W_{j_1},W_{j_2},\hdots,W_{j_t}\}.
\end{equation}
The set $H$ gives the list of undirected edges corresponding to the quantum channels \textit{i.e.} links in the network
\begin{equation}
H=\left\{\,\{A,B\}
\,\,\Bigg|\,\,\parbox{3.6cm}{$\exists$ a quantum channel\\between nodes $A,B\in V$}\right\}.
\end{equation}

We use the integer $\beta_h$ to quantify the total amount of quantum communication in the quantum channel $h\in H$.
More precisely, for $h=\{A,B\}$, if nodes $A$ and $B$ communicate quantum systems of dimensions $\eta_1$, $\eta_2$, $\hdots$, $\eta_\ell$ between them during the erasure correction, then $\beta_h=\eta_1\eta_2\hdots\eta_\ell$.

The definition of entanglement cost of a specific task in a network is somewhat straightforward.
The amount of quantum communication required to do the task can be defined as the entanglement cost.
This is because entanglement distribution and quantum communication become equivalent (due to quantum teleportation) under the assumption that classical communication is free of cost.
Extending this definition, in the multipartite setting, the entanglement cost can be defined for each channel $h\in H$.
Alternatively, we can look at the sum of the entanglement cost across all the channels (or edges) in $H$ in the multipartite setting.
Yamasaki et al~\cite[Section~II-B]{yamasaki17} gives such a definition for \textit{total graph-associated entanglement cost}.
Motivated by the definition in~\cite{yamasaki17}, we define the entanglement cost $\text{EC}_{e,T}(H)$ for quantum erasure correction in a given quantum MDS code as follows.
\begin{equation}
\text{EC}_{e,T}(H)\,:=\,\min_{p\in\mathcal{P}}
\,\,\sum_{h\in H}\log_Q\beta_h.
\end{equation}
Here, the minimization is over the set of all possible erasure correction protocols given by $\mathcal{P}$.
We assume free classical communication among all the $n-1$ non-erased nodes and the replacement node during the erasure correction process.

Physically, $\text{EC}_{e,T}(H)$ indicates the minimum number of $Q$-dimensional qudits needed to perform the quantum communication during erasure correction.
When $Q=2^m$, the quantity $m\,\text{EC}_{e,T}(H)$ indicates the minimum number of qubits communicated or EPR pairs consumed in the network.

\begin{figure*}[t]
\begin{center}
\begin{tikzpicture}[xscale=0.55,yscale=0.55,every node/.style={scale=0.85}]
\node at (-0.7,15) {$W_1$};
\draw (0,15) -- (9.5,15);
\node at (-0.7,14) {$W_2$};
\draw (0,14) -- (9.5,14);
\node at (-0.7,13) {$\vdots$};
\node at (-0.7,12) {$W_t$};
\draw (0,12) -- (9.5,12);
\node at (-0.7,11) {$\widehat{W}_n$};
\draw (0,11) -- (9.5,11);
\node at (-0.7,8.5) {$W_{t+1}$};
\draw (0,8.5) -- (4.2,8.5);
\node at (-0.7,7.5) {$W_{t+2}$};
\draw (0,7.5) -- (4.2,7.5);
\node at (-0.7,6.5) {$\vdots$};
\node at (-0.7,5.5) {$W_{n-1}$};
\draw (0,5.5) -- (4.2,5.5);
\node at (-0.7,4.5) {$W_n$};
\draw (0,4.5) -- (4.2,4.5);
\node at (-0.7,3.5) {$\mathcal{R}$};
\draw (0,3.5) -- (4.2,3.5);
\node at (1.5,16.8) {$\ket{\Phi}_{\mathcal{R}W_1 W_2\hdots W_n}$};
\node at (1.5,16) {$\otimes\ket{\gamma_H}$};
\draw [dashed] (1.5,3) -- (1.5,15.5);
\draw (6.7,8.5) -- (15,8.5);
\node at (11,8.9) {$W_1'$};
\draw (6.7,7.5) -- (15,7.5);
\node at (11,7.9) {$W_2'$};
\node at (11,7) {$\vdots$};
\draw (6.7,5.5) -- (15,5.5);
\node at (11,5.9) {$W_{t-2}'$};
\draw (6.7,4.5) -- (15,4.5);
\node at (11,4.9) {$W_{t-1}'$};
\draw (6.7,3.5) -- (15,3.5);
\node at (11,3.9) {$W_t'$};
\draw (4.2,8.75) rectangle (6.7,3.25);
\node at (5.45,6) {$U_T^{(E)}$};
\draw [dashed] (8,3) -- (8,15.5);
\node at (8,17.2) {$\ket{\psi_{\text{in}}}_{W_1 W_2\hdots W_t}$};
\node at (8.3,16.6) {$_{W_1'W_2'\hdots W_t'}$};
\node at (8,16) {$\otimes\ket{\gamma_H}$};
\draw (9.5,9.75) rectangle (12,15.25);
\node at (10.75,12.5) {$L_2$};
\draw (12,15) -- (21,15);
\node at (21.7,15) {$W_1$};
\draw (12,14) -- (21,14);
\node at (21.7,14) {$W_2$};
\node at (21.7,13) {$\vdots$};
\draw (12,12) -- (21,12);
\node at (21.7,12) {$W_t$};
\draw (12,11) -- (21,11);
\node at (21.7,11) {$\widehat{W}_n$};
\draw (12,10.05) -- (19.05,10.05) -- (19.05,8.9);
\draw (12,9.95) -- (18.95,9.95) -- (18.95,8.9);
\node at (12.7,10.4) {$\ul{v}$};
\draw [dashed] (13.5,3) -- (13.5,15.5);
\node at (13.5,16.6) {$\ket{\psi_{\text{out}}}_{W_1 W_2\hdots W_t\widehat{W}_n}$};
\node at (13.8,16) {$_{W_1'W_2'\hdots W_t'}$};
\draw (15,8.75) rectangle (17.5,3.25);
\node at (16.25,6) {${U_T^{(E)}}^\dagger$};
\draw (17.5,8.5) -- (18.5,8.5);
\draw (18.5,8.1) rectangle (19.5,8.9);
\node at (19,8.5) {$G_1$};
\draw (19.5,8.5) -- (21,8.5);
\node at (21.7,8.5) {$W_{t+1}$};
\draw (18.95,8.1) -- (18.95,7.9);
\draw (19.05,8.1) -- (19.05,7.9);
\draw (17.5,7.5) -- (18.5,7.5);
\draw (18.5,7.1) rectangle (19.5,7.9);
\node at (19,7.5) {$G_2$};
\draw (19.5,7.5) -- (21,7.5);
\node at (21.7,7.5) {$W_{t+2}$};
\draw (18.95,7.1) -- (18.95,6.9);
\draw (19.05,7.1) -- (19.05,6.9);
\node at (19,6.65) {$\vdots$};
\node at (21.7,6.5) {$\vdots$};
\draw (18.95,6.1) -- (18.95,5.9);
\draw (19.05,6.1) -- (19.05,5.9);
\draw (17.5,5.5) -- (18,5.5);
\draw (18,5.1) rectangle (20,5.9);
\node at (19,5.5) {$G_{n-t-1}$};
\draw (20,5.5) -- (21,5.5);
\node at (21.7,5.5) {$W_{n-1}$};
\draw (17.5,4.5) -- (21,4.5);
\node at (21.7,4.5) {$W_n$};
\draw (17.5,3.5) -- (21,3.5);
\node at (21.7,3.5) {$\mathcal{R}$};
\draw [dashed] (20.6,3) -- (20.6,15.5);
\node at (19,16.75) {$\ket{\Phi}_{\mathcal{R}W_1 W_2\hdots W_{n-1}\widehat{W}_n}$};
\node at (20.4,16) {$\otimes\ket{\epsilon}_{W_n}$};
\end{tikzpicture}
\captionsetup{justification=justified}
\caption{LOCC protocol for quantum erasure correction of node $W_n$ from the nodes $W_1$, $W_2$, $\hdots$, $W_t$.
A new node $\widehat{W}_n$ is initialized with resource state $\ket{\gamma_H}$.
We apply and invert the unitary operator $U_T^{(E)}$ acting on the systems $\mathcal{R}$, $W_{t+1}$, $W_{t+2}$, $\hdots$\,, $W_n$, which are not involved in the LOCC protocol $L_2$'s operations.}
\label{fig:pre_post_unitary_locc}
\end{center}
\end{figure*}

\subsection{Erasure correction as LOCC protocol}

We model the erasure correction as a $(t+1)$-partite LOCC protocol $L_2$ involving the $(t+1)$ parties in $W_T \widehat{W}_e$.
This LOCC protocol takes as input a joint state containing
\begin{itemize}
    \item the state $\ket{\Phi}_{\mathcal{R}W_{[n]}}$ obtained by encoding the maximally mixed state $\rho_S$ and
    \item the resource state $\ket{\gamma_H}$ shared among the replacement node and the $t$ helper nodes
\end{itemize}

Consider an $[[n,2t-n]]_Q$ quantum MDS code.
Let $\mathcal{S}$ be the system which contains the maximally mixed state\begin{equation}
\rho_\mathcal{S}=\frac{1}{Q^{2t-n}}\sum_{\ul{s}\in\A_Q^{2t-n}}\ket{\ul{s}}\bra{\ul{s}}.
\end{equation}
Take a reference system $\mathcal{R}$ such that joint system $\mathcal{RS}$ holds the pure state
\begin{equation}
\ket{\phi}_\mathcal{RS}=\frac{1}{\sqrt{Q^{2t-n}}}\sum_{\ul{s}\in\A_Q^{2t-n}}\ket{\ul{s}}_\mathcal{R}\ket{\ul{s}}_\mathcal{S}.
\end{equation}
After encoding the state $\rho_\mathcal{S}$, the state of the system $\mathcal{R}W_{[n]}$ is
\begin{equation}
\label{eq:qmds_enc_1}
\ket{\Phi}_{\mathcal{R}W_{[n]}}=\,(I_\mathcal{R}\,\otimes\,E_{\mathcal{SA}})
\,.\left(\ket{\phi}_{\mathcal{RS}}\,\ket{0}^{\otimes(2n-2t)}_{\mathcal{A}}\right).
\end{equation}
Here $E$ is an encoding operator for the $[[n,2t-n]]_Q$ quantum MDS code as given in Eq.~\eqref{eq:enc_basis_state}.

We can use the idea of quantum teleportation to formulate quantum erasure correction as a circuit using an $n$-partite LOCC protocol.
This is done by using a resource state to replace the quantum communication in each channel $h=\{h(1),h(2)\}\in H$ with a maximally entangled state.
\begin{equation}
\ket{\gamma_H}\,\,:=\,\,\bigotimes_{h\in H}\,\ket{\beta_h^+}_{h(1),h(2)}
\end{equation}
See Fig.~\ref{fig:eras_corr_as_locc}~{\!(b)} for an illustration with $T=[t]$ and $e=n$, .

The following theorem modifies this circuit to another circuit which uses a $(t+1)$-partite LOCC protocol as illustrated in Fig.~\ref{fig:eras_corr_as_locc}~{\!(c)}.
\begin{theorem}
\label{th:ec_locc}
Consider an $[[n,k]]_Q$ quantum code.
Let $e\in[n]$, $T=\{j_1,j_2,\hdots,j_t\}\subseteq[n]\backslash e$ and $[n]\backslash(T\cup\{e\})=\{j_{t+1},j_{t+2},\hdots,j_{n-1}\}$.
Let $H$ denote a set of quantum channels between pairs of parties in $\{\widehat{W}_e,W_{j_1},W_{j_2},\hdots,W_{j_t}\}$.
The erasure of the physical qudit in node $W_e$ can be corrected with quantum communication of dimension $\beta_h$ in channel $h$ for every $h\in H$ only if an LOCC$(W_{j_1},W_{j_2},\hdots,W_{j_t},\widehat{W}_e)$ protocol $L_2$ exists such that the operation
\begin{flalign}
&\,\,\,\,(L_2)_{W_T\widehat{W}_n}\otimes (G_1)_{W_{j_{t+1}}}\otimes\hdots\otimes(G_{n-t-1})_{W_{j_{n-1}}}\otimes\idmtx_{W_e\mathcal{R}}\,\,\,\,
\end{flalign}
takes the initial state
\begin{equation}
\ket{\Phi}_{\mathcal{R}W_1 W_2\hdots W_n}
\,\,\bigotimes_{h\in H}\,\ket{\beta_h^+}_{h(1),h(2)}
\end{equation}
to the final state
\begin{equation}
\ket{\Phi}_{\mathcal{R}W_1 W_2\hdots W_{e-1}\widehat{W}_e W_{e+1} W_{e+2}\hdots W_n}
\,\,\otimes\,\ket{\epsilon}_{W_e}
\end{equation}
where $\ket{\epsilon}$ is some fixed pure state in $\Hs_Q$ and $G_1$, $G_2$, $\hdots$~, $G_{n-t-1}$ are unitary operators over $\Hs_Q$ dependent on the local measurement outcomes from $L_2$.
\end{theorem}
\begin{proof}
The proof is by construction.
Assume that an erasure correction protocol $\mathcal{P}$ for node $W_e$ with quantum communication of dimension $\beta_j$ for every $j\in T$ exists.
Then we can do the erasure correction using an LOCC($W_1$, $W_2$, $\hdots$~, $W_{e-1}$, $W_{e+1}$, $W_{e+2}$, $\hdots$~, $W_n$, $\widehat{W}_n$) protocol $L_1$ by replacing the quantum communications over the channels in $H$ with quantum teleportations using the resource state $\ket{\gamma_H}$.
This is possible due to Lemma~\ref{lm:q_tele}.

From \cite[Section~II]{grassl22}, we know that each physical qudit of any logical state in the $[[n,2t-n]]_Q$ QMDS code is maximally entangled.
This implies that each node in the set $\{W_j\}_{j\notin T\cup\{e\}}$ before and after the erasure correction has the same Schmidt coefficients given by
\begin{equation}
\left(\frac{1}{\sqrt{Q}},\frac{1}{\sqrt{Q}},\hdots,\frac{1}{\sqrt{Q}}\right)\in\mathbb{R}^Q.
\end{equation}
Now, by successive application of Lemma~\ref{lm:locc_ent_unchanged} for the $n-t-1$ parties in $\{W_j\}_{j\notin T\cup\{e\}}$, we obtain the LOCC protocol $L_2$ and the unitary operators $G_1$, $G_2$, $\hdots$~, $G_{n-t-1}$.
\end{proof}

We further modify the circuit with the following two operations.
The unitary operator $U_T^{(E)}$ is applied to the system $\mathcal{R}W_{[n]\backslash T}$ (along with suitable ancillas) just before the LOCC protocol $L_2$.
Immediately after the execution of protocol $L_2$, the inverse operator ${U_T^{(E)\dagger}}$ is applied.
This modification does not affect the success or failure of the erasure correction.
The resultant circuit is illustrated in Fig.~\ref{fig:pre_post_unitary_locc}.

Without loss of generality, we take $e=n$ and $T=[t]$ here onwards.
For ease of notation, we refer to the $t$ qudits at the output of $U_T^{(E)}$ as belonging to individual systems $W_1'$, $W_2'$, $\hdots$\, $W_t'$.
For any set $L=\{i_1,i_2,\hdots,i_\ell\}\subseteq[t]$, the joint system $W_{i_1}'W_{i_2}'\hdots W_{i_\ell}'$ is denoted as $W_L'$.

\subsection{The state before the LOCC protocol $L_2$}

Let us derive the expression for the state $\ket{\psi_\text{in}}$ of the overall system $\mathcal{R}W_TW_T'\widehat{W}_n$ (excluding the resource state) just before the execution of $L_2$.
By Lemma~\ref{lm:unit_qmds}, the state $\ket{\Phi}_{\mathcal{R}W_{[n]}}$ can be written as
\begin{flalign}
&\,\,\frac{1}{\sqrt{Q^t}}
\sum_{\ul{s}\in\A_Q^{2t-n}}\sum_{\substack{r_1,r_2,\\\hdots,r_{n-t}\\\in\,\,\A_Q}}
\ket{\ul{s}}_\mathcal{R}\,U_T^{(E)}\!\left(\ket{\ul{s}}\ket{r_1}\!\hdots\ket{r_{n-t}}\right)_{W_T}&
\nonumber
\\*[-0.7cm]&\hspace{3.35cm}\,\ket{r_1}_{W_{t+1}}\ket{r_2}_{W_{t+2}}\hdots\ket{r_{n-t}}_{W_n}\!.\!\!\!\!\!&
\label{eq:phi_r_wn}
\end{flalign}
After applying $U_T^{(E)}$ in the subsystem $\mathcal{R}W_{[t+1,n]}$ of the state $\ket{\Phi}_{\mathcal{R}W_{[n]}}$, we obtain
\begin{flalign}
&\ket{\psi_\text{in}}&\nonumber
\\&\,\,=\frac{1}{\sqrt{Q^t}}
\sum_{\ul{s}\in\A_Q^{2t-n}}
\,\sum_{\substack{r_1,r_2,\hdots,\\r_{n-t}\in\A_Q}}
\,U_T^{(E)}\!\left(\ket{\ul{s}}\ket{r_1}\!\hdots\ket{r_{n-t}}\right)_{W_T}&
\nonumber
\\*[-0.55cm]&\hspace{4.2cm}U_T^{(E)}\!\left(\ket{\ul{s}}\ket{r_1}\!\hdots\ket{r_{n-t}}\right)_{W_T'}\!\!\!\!\!\!\!\!\!&
\\&\,\,=\frac{1}{\sqrt{Q^t}}
\sum_{\substack{m_1,m_2,\\\hdots,m_t\in\A_Q}}
\,(\ket{m_1}\ket{m_2}\hdots\ket{m_t})_{W_T}&
\nonumber
\\*[-0.55cm]&\hspace{4.2cm}(\ket{m_1}\ket{m_2}\hdots\ket{m_t})_{W_T}.&
\\&\,\,=
\!\left(\!\frac{1}{\sqrt{Q}}\sum_{\substack{m_1\\\in\A_Q}}\ket{m_1}_{W_1}\ket{m_1}_{W_1'}\!\right)
\!\left(\!\frac{1}{\sqrt{Q}}\sum_{\substack{m_2\\\in\A_Q}}\ket{m_2}_{W_2}\ket{m_2}_{W_2'}\!\right)\nonumber
\\*&\hspace{4.15cm}\hdots
\left(\!\frac{1}{\sqrt{Q}}
\sum_{\substack{m_t\\\in\A_Q}}\ket{m_t}_{W_t}\ket{m_t}_{W_t'}\!\right)\!.&\nonumber
\\*\label{eq:psi_in}
\end{flalign}
Clearly, for any $j\in[t]$, the joint system $W_jW_j'$ in the state $\ket{\psi_\text{in}}$ holds a pure state.
This leads to the following lemma.
\begin{lemma}
\label{lm:sr_psi_in}
For the state $\ket{\psi_\text{in}}_{W_T W_T'}$ as given in Eq.~\eqref{eq:psi_in},
\begin{eqnarray}
\sr(W_j W_j')_{\psi_\text{in}}=1
\end{eqnarray}
for any $j\in[t]$.
\end{lemma}

\subsection{The state after the LOCC protocol $L_2$}

Let $\ket{\psi_\text{out}}$ denote the state of the system $\mathcal{R}W_TW_T'\widehat{W}_n$ immediately after the execution of $L_2$.
The expression for this state can be written by backtracing the steps in the circuit from the final erasure-corrected state.

The final erasure corrected state $\ket{\Phi}_{\mathcal{R}W_1 W_2\hdots W_{n-1}\widehat{W}_n}\otimes\ket{\epsilon}_{W_n}$ is given by
\begin{flalign}
&\frac{1}{\sqrt{Q^t}}\sum_{\substack{\ul{s}\in\\\A_Q^{2t\!-\!n}}}
\,\sum_{\substack{r_1,r_2,\\\hdots,\\r_{n-t}\\\in\,\,\A_Q}}
\,\ket{\ul{s}}_\mathcal{R}\,\,U_T^{(E)}\!\left(\ket{\ul{s}}\ket{r_1}\hdots\ket{r_{n-t}}\right)_{W_T}
\,\,\ket{r_{n-t}}_{\widehat{W}_n}&
\nonumber
\\*[-0.9cm]&\hspace{2.6cm}\ket{r_1}_{W_{t+1}}\ket{r_2}_{W_{t+2}}\hdots\ket{r_{n-t-1}}_{W_{n-1}}
\ket{\epsilon}_{W_n}\!.&
\nonumber
\\*[0.1cm]\label{eq:final_correct_state}
\end{flalign}
The state $\ket{\psi_\text{out}}$ immediately after the execution of $L_2$ can be obtained by suitably applying the operators $G_1^\dagger$, $G_2^\dagger$,\,$\hdots$\,, $G_{n-t-1}^\dagger$ and then $U_T^{(E)}$ to the state in Eq.~\eqref{eq:final_correct_state}.
The expression for $\ket{\psi_\text{out}}$ thus obtained is
\begin{flalign}
&\frac{1}{\sqrt{Q^t}}
\!\sum_{\substack{\ul{s}\,\in\\\A_Q^{2t\!-\!n}}}\sum_{\substack{r_1,r_2,\\\hdots,\\r_{n-t}\\\in\,\,\A_Q}}
\,U_T^{(E)}\!\left(\ket{\ul{s}}\ket{r_1}\ket{r_2}\hdots\ket{r_{n-t}}\right)_{W_T}
\,\,\ket{r_{n-t}}_{\widehat{W}_n}&
\nonumber
\\*[-0.9cm]&\hspace{2.45cm}U_T^{(E)}
\big(\!\ket{\ul{s}}G_1^\dagger\!\ket{r_1}\hdots G_{n-t-1}^\dagger\!\ket{r_{n-t-1}}\!\ket{\epsilon}\!\big)_{W_T'}.&
\nonumber
\\*[0.1cm]\label{eq:psi_out}
\end{flalign}

The following lemma gives the Schmidt ranks $\sr(W_jW_j')$ and $\sr(\widehat{W}_n)$ for the state $\ket{\psi_\text{out}}$.
\begin{lemma}
\label{lm:sr_psi_out}
For the state $\ket{\psi_\text{out}}_{W_T\widehat{W}_n W_T'}$ as given in Eq.~\eqref{eq:psi_out},
\begin{eqnarray}
\sr(\widehat{W}_n)_{\psi_\text{out}}&=&Q,
\label{eq:sr_rep_node}
\\\sr(W_jW_j')_{\psi_\text{out}}&=&Q^2
\label{eq:sr_psi_out}
\end{eqnarray}
for any $j\in[t]$.
\end{lemma}
\begin{proof}
See Appendix~\ref{ap:post_locc_sr}.
\end{proof}

\section{Entanglement cost}
\label{s:ent_cost}

In this section, we use the monotonicity of Schmidt rank under LOCC to obtain a lower bound on the entanglement cost of two different star networks (Theorem~\ref{th:ent_cost_bound_1} and Theorem~\ref{th:ent_cost_bound_2}).
For this, we compare the Schmidt ranks from the previous section.
Then we prove that these bounds are optimal (Corollary~\ref{co:ent_cost_1} and Corollary~\ref{co:ent_cost_2}).
This is by showing that the simple method of downloading the qudits from the helper nodes to the hub node, generating the replacement qudit and then sending the downloaded qudits back achieves the lower bound.

We consider the two possible star network topologies for the quantum network involved in erasure correction.
Recall that, for $T=\{j_1,j_2,\hdots,j_t\}$ the nodes in the quantum network, are given by $V=\{\widehat{W}_e,W_{j_1},W_{j_2},\hdots,W_{j_t}\}$.
The first topology is the star topology in which the replacement node is the hub.
\begin{equation}
H_1:=\left\{\{\widehat{W}_e,W_{j_1}\},\{\widehat{W}_e,W_{j_1}\},\hdots,\{\widehat{W}_e,W_{j_t}\}\right\}
\end{equation}
The resource state corresponding to this network can be written as
\begin{equation}
\ket{\gamma_{H_1}}\,\,=\,\,\bigotimes_{j\in T}\,\left|\beta_{\{\widehat{W}_e,W_j\}}^+\right>_{\widehat{W}_e,W_j}.
\end{equation}

The following theorem and corollary give the entanglement cost of quantum erasure correction using the network given by $H_1$.

\begin{theorem}
\label{th:ent_cost_bound_1}
Consider an $[[n,2t-n]]_Q$ quantum MDS code with $t>1$.
Let $e\in[n]$ and $T\subseteq[n]\backslash e$ such that $|T|=t$.
The erasure of node $W_e$ can be corrected with quantum communication over the star network $H_1$ with the replacement node as the hub only if
\begin{equation}
\log_Q\beta_{\{\widehat{W}_e,W_j\}}\geq 2
\end{equation}
for every $j\in T$.
\end{theorem}
\begin{proof}
If erasure correction is possible, then $\ket{\psi_\text{out}}$ can be obtained by applying the LOCC protocol $L_2$ to the state $\ket{\psi_\text{in}}\otimes\ket{\gamma_{H_1}}$.
Then, from the monotonicity of Schmidt rank under LOCC and Lemmas~\ref{lm:sr_psi_in} and \ref{lm:sr_psi_out}, for any $j\in T$,
\begin{eqnarray}
\sr\left(W_j W_j'\right)_{\psi_\text{in}}
\sr\left(W_j W_j'\right)_{\gamma_{H_1}}
&\!\!\!\geq\!\!\!&
\sr\left(W_j W_j'\right)_{\psi_\text{out}}
\,\,\,\,\,\,\,\,
\\1\,.\,\beta_{\{\widehat{W}_e,W_j\}}
&\!\!\!\geq\!\!\!&
Q^2
\\\log_Q\beta_{\{\widehat{W}_e,W_j\}}
&\!\!\!\geq\!\!\!&
2.
\end{eqnarray}
\end{proof}

\begin{corollary}
\label{co:ent_cost_1}
Consider an $[[n,2t-n]]_Q$ quantum MDS code.
The entanglement cost of correcting a single physical qudit using a quantum network $H_1$ is given by
\begin{equation}
\text{EC}_{e,T}(H_1)=2t.
\end{equation}
\begin{proof}
Applying Theorem~\ref{th:ent_cost_bound_1}, we obtain
\begin{equation}
\sum_{h\in H_1}\log_Q\beta_h\geq 2t.\label{eq:ec_lb_1}
\end{equation}
From Lemma~\ref{lm:t_is_suff}, we know that through quantum communication with nodes $W_{j_1}$, $W_{j_2}$, $\hdots$, $W_{j_t}$, the replacement node $\widehat{W}_e$ can perform the erasure correction.
This can be done simply by node $\widehat{W}_e$ downloading the $t$ qudits from the $t$ helper nodes, performing the operations required for erasure correction locally and sending back the $t$ qudits.
This download-and-return protocol is illustrated in Fig.~\ref{fig:network_diag_1}.
The quantum communication needed for this erasure correction protocol is $2t$ qudits, which is the same as the lower bound obtained in Eq.~\eqref{eq:ec_lb_1}.
Hence, the entanglement cost of quantum erasure correction in the network $H_1$ is equal to $2t$ qudits (each of dimension $Q$).
\end{proof}
\end{corollary}

In the second star topology, one of the helper nodes (say $W_{j_1}$) is the hub.
\begin{eqnarray}
H_2:=\left\{
\begin{array}{c}
\hspace{-1.4cm}
\{W_{j_1},\widehat{W}_e\},\{W_{j_1},W_{j_2}\},
\\*[0.1cm]\hspace{0.4cm}\{W_{j_1},W_{j_3}\},\hdots,\{W_{j_1},W_{j_t}\}
\end{array}
\!\!\!\right\}
\end{eqnarray}
The resource state corresponding to the network $H_2$ can be written as
\begin{equation}
\ket{\gamma_{H_2}}\,\,=\,\,
\Big|\beta_{\{W_{j_1}\widehat{W}_e\}}^+\Big>_{\!W_{j_1},\widehat{W}_e}
\,\,\bigotimes_{j\in T\setminus\{j_1\}}\Big|\beta_{\{W_{j_1},W_j\}}^+\Big>_{\!W_{j_1},W_j}.
\end{equation}

The following theorem and corollary give the entanglement cost of quantum erasure correction using the network $H_2$.
See Fig.~\ref{fig:network_diag_2} for an illustration of the quantum communications involved in the erasure correction using this network.

\begin{theorem}
\label{th:ent_cost_bound_2}
Consider an $[[n,2t-n]]_Q$ quantum MDS code.
Let $e\in[n]$, $T\subseteq[n]\backslash e$ such that $|T|=t$ and $j_1\in T$.
The erasure of node $W_e$ can be corrected with quantum communication over the star network $H_2$ with a helper node $W_{j_1}$ as the hub only if
\begin{equation}
\log_Q\beta_{\{W_{j_1},W_j\}}\geq 2
\end{equation}
for every $j\in T\setminus\{j_1\}$ and
\begin{equation}
\log_Q\beta_{\{W_{j_1},\widehat{W}_e\}}\geq 1.
\end{equation}
\end{theorem}
\begin{proof}
If erasure correction is possible, then $\ket{\psi_\text{out}}$ can be obtained by applying the LOCC protocol $L_2$ to the state $\ket{\psi_\text{in}}\otimes\ket{\gamma_{H_2}}$.
Then, from the monotonicity of Schmidt rank under LOCC and Lemmas~\ref{lm:sr_psi_in} and \ref{lm:sr_psi_out}, for any $j\in T\setminus\{j_1\}$,
\begin{eqnarray}
\sr\left(W_j W_j'\right)_{\psi_\text{in}}
\sr\left(W_j W_j'\right)_{\gamma_{H_2}}
&\!\!\!\geq\!\!\!&
\sr\left(W_j W_j'\right)_{\psi_\text{out}}
\,\,\,\,\,\,\,\,
\\1\,.\,\beta_{\{W_{j_1},W_j\}}
&\!\!\!\geq\!\!\!&
Q^2
\\\log_Q\beta_{\{W_{j_1},W_j\}}
&\!\!\!\geq\!\!\!&
2.
\end{eqnarray}
From Eq.~\eqref{eq:sr_rep_node} in Lemma~\ref{lm:sr_psi_out}, we obtain
\begin{eqnarray}
\sr(\widehat{W}_e)_{\psi_\text{in}}
\sr(\widehat{W}_e)_{\gamma_{H_2}}
&\!\!\!\geq\!\!\!&
\sr(\widehat{W}_e)_{\psi_\text{out}}
\,\,\,\,\,\,\,\,
\\1\,.\,\beta_{\{W_{j_1}\widehat{W}_e\}}
&\!\!\!\geq\!\!\!&
Q
\\\log_Q\beta_{\{W_{j_1}\widehat{W}_e\}}
&\!\!\!\geq\!\!\!&
1.
\end{eqnarray}
\end{proof}


\begin{corollary}
\label{co:ent_cost_2}
Consider an $[[n,2t-n]]_Q$ quantum MDS code.
The entanglement cost of correcting a single physical qudit using the quantum network $H_2$ is given by
\begin{equation}
\text{EC}_{e,T}(H_2)=2t-1.
\end{equation}
\end{corollary}
\begin{proof}
Applying Theorem~\ref{th:ent_cost_bound_2}, we obtain
\begin{equation}
\sum_{h\in H_2}\log_Q\beta_h\geq 2t-1.\label{eq:ec_lb_2}
\end{equation}
From Lemma~\ref{lm:t_is_suff}, we know that through quantum communication with nodes $W_{j_1}$, $W_{j_2}$, $\hdots$, $W_{j_t}$, the replacement node $\widehat{W}_e$ can perform the erasure correction.
This can be done in network $H_2$ by the hub node $W_{j_1}$ downloading the $t-1$ qudits from the other $t-1$ helper nodes, performing the operations required for erasure correction locally and sending back the $t-1$ qudits to respective helper nodes and the ancilla qudit used to perform the erasure correction (see proof of Lemma~\ref{lm:t_is_suff}) to the replacement node.
This protocol, which is similar to the download-and-return protocol from the proof of Corollary~\ref{co:ent_cost_1}, is illustrated in Fig.~\ref{fig:network_diag_2}.
The quantum communication needed for this erasure correction protocol is $2t-1$ qudits, which is the same as the lower bound obtained in Eq.~\eqref{eq:ec_lb_2}.
Hence, the entanglement cost of quantum erasure correction in the network $H_2$ is equal to $2t-1$ qudits (each of dimension $Q$).

\end{proof}

\section{Conclusion}
\label{s:conc}

In this paper, we derived the entanglement cost of erasure correction in quantum MDS codes using star networks.
We looked at the case when the number of helper nodes is minimal ($d=t$).
We modeled the erasure correction as a circuit with an LOCC protocol and then applied the monotonicity of Schmidt rank to derive lower bounds on entanglement cost.
The important idea we used was to apply and invert a suitable unitary operation on nodes not directly involved in the erasure correction.
Then we showed that these bounds are optimal and hence give the expressions for the entanglement cost.
We observe that the entanglement cost can be achieved by the simple download-and-return erasure correction protocols.

As future work, we can look at the entanglement cost while using a non-minimal number of helper nodes ($d>t$).
We can also study the entanglement cost of more general network topologies and code families.
Another direction is to explore the case of approximate quantum erasure correction.
Whether our results can be expanded to the case of distributed quantum \textit{error} correction will be a very interesting problem.
Note that, unlike erasures, we do not know the location of the physical qudits facing errors.

\section*{Acknowledgements}
This work was supported by funding from the NWO Zwaartekracht Grant QSC and an NWO VICI Grant.
The author thanks \'Alvaro G. I\~nesta and Bethany Davies for critical feedback on the manuscript.

\bibliographystyle{IEEEtran}
\balance
\bibliography{refs-dqs}
\onecolumn
\appendices

\counterwithin{lemma}{section}
\counterwithin{definition}{section}

\section{Properties of LOCC protocols}
\label{ap:locc}

Here, we give the proof for Lemma~\ref{lm:locc_ent_unchanged}.
To do this, we first need to define \textit{majorization}.
\begin{definition}
Let $\ul{\sigma}=(\sigma_1,\sigma_2,\hdots,\sigma_\ell)\in\mathbb{R}^\ell$ and $\ul{\mu}=(\mu_1,\mu_2,\hdots,\mu_\ell)\in\mathbb{R}^\ell$ be two vectors such that
\begin{eqnarray}
\sigma_1\geq\sigma_2\geq\hdots\geq\sigma_\ell\geq 0\,,
\\\mu_1\geq\mu_2\geq\hdots\geq\mu_\ell\geq 0\,.
\end{eqnarray}
We say that $\ul{\sigma}$ majorizes $\ul{\mu}$ (denoted as $\ul{\sigma}\succ\ul{\mu}$) if
\begin{equation}
\sum_{i=1}^{k}\sigma_i\geq\sum_{i=1}^{k}\mu_i
\end{equation}
for every $k\in[\ell]$.
\end{definition}

The following lemma (commonly known as Nielsen's theorem) gives a necessary and sufficient condition for a pure state to be transformed through bipartite LOCC to another pure state.
\begin{lemma}\cite[Theorem~1]{nielsen99}
\label{lm:nielsen_theorem}
A bipartite state $\ket{\phi_\text{in}}_{AB}$ can be transformed to another bipartite state $\ket{\phi_\text{out}}_{AB}$ with local operations and classical communication if and only if $\ul{\lambda}(A)_{\phi_\text{out}}\succ\ul{\lambda}(A)_{\phi_\text{in}}$.
\end{lemma}

\newtheorem*{orig_locc_lemma_ent_unchanged}{Lemma~\ref{lm:locc_ent_unchanged}}
\begin{orig_locc_lemma_ent_unchanged}
Let $\ket{\psi_\text{in}}_{P_1P_2\hdots P_m}$ and $\ket{\psi_\text{out}}_{P_1P_2\hdots P_m}$ be two $m$-partite pure states such that
\begin{equation}
\ul{\lambda}(P_m)_{\psi_\text{out}}=\ul{\lambda}(P_m)_{\psi_\text{in}}
=\left(\frac{1}{\sqrt{Q}},\frac{1}{\sqrt{Q}},\hdots,\frac{1}{\sqrt{Q}}\right)\in\mathbb{R}^Q.
\end{equation}
An LOCC$(P_1,P_2,\hdots,P_m)$ protocol $L_1$ which transfroms $\ket{\psi_\text{in}}$ to $\ket{\psi_\text{out}}$ exists if and only if an LOCC$(P_1,P_2,\hdots,P_{m-1})$ protocol $L_2$ exists such that $\ket{\psi_\text{out}}$ can be obtained by
\begin{enumerate}
\item applying $(L_2)_{P_1P_2\hdots P_{m-1}}\,\otimes\,\idmtx_{P_m}$ on $\ket{\psi_\text{in}}$ to obtain the state $\ket{\psi_\text{mid}}$ and
\item applying $\idmtx_{P_1P_2\hdots P_{m-1}}\,\otimes\,G(\ul{v})_{P_m}$ on $\ket{\psi_\text{mid}}$.
Here $G(\ul{v})$ is a unitary operator dependent on the local measurements $\ul{v}$ obtained from the execution of $L_2$ in the previous step.
\end{enumerate}
\end{orig_locc_lemma_ent_unchanged}
\begin{proof}
If part of the statement is proved by taking $L_1$ to be the combined protocol involving steps 1 and 2.
The proof for the only if part is by construction as given below.

Any local operation happening in $P_m$ can be done with either a unitary operator or a projective measurement (with suitable ancillas).
We show below that any local projective measurement in $P_m$ during the execution of $L_1$ can be replaced with a suitable unitary operator.

Consider any local projective measurement done by $P_m$ during the protocol $L_1$'s execution while transforming $\ket{\psi_\text{in}}$ to $\ket{\psi_\text{out}}$.
Let $\ket{\phi_\text{in}}_{P_1P_2\hdots P_m}$ and $\ket{\phi_\text{out}}_{P_1P_2\hdots P_m}$ be the intermediate states just before and after the measurement.
From Lemma~\ref{lm:nielsen_theorem}, we know that
\begin{equation}
\ul{\lambda}(P_m)_{\psi_\text{out}}\,\succ\,
\ul{\lambda}(P_m)_{\phi_\text{out}}\,\succ\,
\ul{\lambda}(P_m)_{\phi_\text{in}}\,\succ\,
\ul{\lambda}(P_m)_{\psi_\text{in}}.
\end{equation}
However, since $\ul{\lambda}(P_m)_{\psi_\text{out}}=\ul{\lambda}(P_m)_{\psi_\text{out}}=\left(\frac{1}{\sqrt{Q}},\frac{1}{\sqrt{Q}},\hdots,\frac{1}{\sqrt{Q}}\right)\in\mathbb{R}^Q$, we get
\begin{equation}
\ul{\lambda}(P_m)_{\phi_\text{out}}=
\ul{\lambda}(P_m)_{\phi_\text{in}}=\left(\frac{1}{\sqrt{Q}},\frac{1}{\sqrt{Q}},\hdots,\frac{1}{\sqrt{Q}}\right)\in\mathbb{R}^Q.
\label{eq:schmidt_coeff_unchnaged}
\end{equation}
Let the Schmidt decomposition of $\ket{\phi_\text{in}}$ between the subsystems $P_m$ and $P_1P_2\hdots P_{m-1}$ be
\begin{equation}
\ket{\phi_\text{in}}=\frac{1}{\sqrt{Q}}
\,\sum_{j=1}^{Q}\,
\ket{a_j}_{P_m}\ket{b_j}_{P_1 P_2\hdots P_{m-1}}.
\end{equation}
Assume that the measurement outcome $v$ is obtained after the local measurement.
Denote the projector corresponding to the outcome $v$ as $V$.
Then the state $\ket{\phi_\text{out}}$ is given by
\begin{eqnarray}
\ket{\phi_\text{out}}=\frac{V_{P_m}\otimes\idmtx_{P_1 P_2\hdots P_{m-1}}
\ket{\phi_\text{in}}}{\sqrt{\bra{\phi_\text{in}}V_{P_m}\otimes\idmtx_{P_1 P_2\hdots P_{m-1}}
\ket{\phi_\text{in}}}}
=\sum_{j\in J}
\mu_j\ket{\widehat{a}_j}_{P_m}\ket{b_j}_{P_1 P_2\hdots P_{m-1}}
\end{eqnarray}
where $J=\{j\in[Q]:\bra{a_j}V\!\ket{a_j}\neq 0\}\subseteq[Q]$ and
\begin{eqnarray}
\mu_j=\sqrt{\frac{\bra{a_j}V\!\ket{a_j}}{\sum_{\ell\in J}
\bra{a_\ell}V\!\ket{a_\ell}}}\,,
\hspace{1cm}
\ket{\widehat{a}_j}=\frac{V\!\ket{a_j}}{\bra{a_j}V\!\ket{a_j}}
\hspace{1cm}
\forall j\in J.
\end{eqnarray}
Now, the state held by the party $P_m$ alone can be described with the density matrix
\begin{equation}
\rho(P_m)_{\phi_\text{out}}
=\sum_{j\in J}\,\mu_j^2\ket{\widehat{a}_j}\bra{\widehat{a}_j}.
\end{equation}
Applying \cite[Theorem~11.10]{nielsen10}, we obtain
\begin{equation}
\mathbf{S}(P_m)_{\phi_\text{out}}
\,\,\leq\,\,H\big(\{\mu_j\}_{j\in J}\big)
\label{eq:app_phi_out_entr_ub}
\end{equation}
where equality is achieved if and only if $\left\{\ket{\widehat{a}_i}\right\}_{j\in J}$ is a set of orthornormal states.
However, from Eq.~\eqref{eq:schmidt_coeff_unchnaged}, we know that
\begin{equation}
\mathbf{S}(P_m)_{\phi_\text{out}}
=\log Q.
\label{eq:app_phi_out_entr_eq}
\end{equation}
Taken together, Eq.~\eqref{eq:app_phi_out_entr_ub} and \eqref{eq:app_phi_out_entr_eq} imply that
\begin{equation}
H\big(\{\mu_j\}_{j\in J}\big)=\log Q.
\end{equation}
This in turn implies that $J=[Q]$ and $\left\{\ket{\widehat{a}_j}\right\}_{j\in Q}$ is a set of $Q$ orthornormal states.

\textit{Constructing $L_2$ from $L_1$.}
Take some unitary operator $g_v$ which maps $\ket{a_j}$ to $\ket{\widehat{a}_j}$ (with suitable ancillas) for every $j\in[Q]$.
Let $p(v)$ be the probability of getting a measurement outcome $v$ and $\Lambda$ be the set of all possible measurement outcomes.
Then we can replace the local measurement by $P_m$ with
\begin{enumerate}
\item some other party (say $P_1$) locally generating and measuring in the computational basis the state
$\sum_{i\in\Lambda}\sqrt{p(i)}\ket{i}$ and
\item based on that measurment outcome $v'$, party $P_m$ applying the unitary operator $g_{v'}$.
\end{enumerate}
By doing the above replacement for every local measurement by $P_m$ in LOCC$(P_1,P_2,\hdots,P_m)$ protocol $L_1$, we obtain the LOCC$(P_1,P_2,\hdots,P_{m-1})$ protocol $L_2$ and the unitary operator $G$.
Specifically, the unitary operator $G$ is obtained by combining the operators $\{g_{v'}\}$ over all the local measurements in $P_m$ and other unitary operators in $P_m$ during execution of $L_1$.
\end{proof}

\section{Properties of quantum MDS codes}
\label{ap:qmds}

\newtheorem*{orig_lemma_unit_qmds}{Lemma~\ref{lm:unit_qmds}}
\begin{orig_lemma_unit_qmds}
Consider an $[[n,2t-n]]_Q$ quantum MDS code with an encoding operator $E$.
For any $T\subseteq[n]$ of size $|T|=t$, there exists a unitary operator $U_T^{(E)}$ over $\Hs_Q^{\otimes t}$ such that the encoding is given by
\begin{equation}
\label{ap_eq:unit_qmds}
E\left(\ket{\ul{s}}\ket{0}^{\otimes(2n-2t)}\right)
\,\,=\,\,
\frac{1}{\sqrt{Q^{n-t}}}\!\sum_{\substack{r_1,r_2,\hdots,\\r_{n-t}\,\in\,\A_Q}}
U_T^{(E)}\!\left(\ket{\ul{s}}\ket{r_1}\!\ket{r_2}\hdots\ket{r_{n-t}}\right)_{W_T}
\,\,(\ket{r_1}\!\ket{r_2}\hdots\ket{r_{n-t}})_{W_{[n]\setminus{D}}}
\end{equation}
for all $\ul{s}\in\A_Q^{2t-n}$.
\end{orig_lemma_unit_qmds}

\begin{proof}
Let the logical state obtained by encoding the state $\ket{\ul{s}}\in\Hs_Q^{2t-n}$ be $\ket{\phi_\ul{s}}$.
From \cite[Observation~6]{huber20}, we know that any set of $n-t$ physical qudits in a code state of quantum MDS code is in a maximally mixed state.
This implies
\begin{equation}
\rho(W_{[n]\setminus T})_{\phi_\ul{s}}=\frac{1}{Q^{n-t}}\sum_{\ul{r}\in\A_Q^{n-t}}\ket{\ul{r}}\bra{\ul{r}}.
\end{equation}
Hence, by Schmidt decomposition, we write
\begin{equation}
\ket{\phi_\ul{s}}=\frac{1}{\sqrt{Q^{n-t}}}
\,\sum_{\ul{r}\in\A_Q^{n-t}}
\,\ket{\alpha_{\ul{s},\ul{r}}}_{W_T}
\,\ket{\ul{r}}_{W_{[n]\setminus T}}.
\label{eq:psi_s_schmidt}
\end{equation}
Taking $U_T$ to be the unitary operator over $\Hs_Q^{\otimes t}$ which maps $\ket{\ul{s}}\ket{\ul{r}}$ to $\ket{\alpha_{\ul{s},\ul{r}}}$, we obtain the result.
\end{proof}

\newtheorem*{orig_lemma_unit_p1_qmds}{Lemma~\ref{lm:unit_p1_qmds}}
\begin{orig_lemma_unit_p1_qmds}
Let $U_T^{(E)}$ be the unitary operator as defined in Lemma~\ref{lm:unit_qmds}.
The $[[t+1,t-1]]_Q$ quantum code with the encoding
\begin{eqnarray}
\label{eq:enc_unit_p1_qmds}
\ket{s_1}\ket{s_2}\hdots\ket{s_{t-1}}
\,\mapsto\,\frac{1}{\sqrt{Q}}
\sum_{r\in\A_Q}U_T^{(E)}\big(\!\ket{s_1}\ket{s_2}\hdots\ket{s_{t-1}}\ket{r}\!\big)
\,\ket{r}
\hspace{0.8cm}\forall (s_1,s_2,\hdots,s_{2t-d-1})\in\A_Q^{2t-d-1}
\end{eqnarray}
is a quantum MDS code.
\end{orig_lemma_unit_p1_qmds}

\begin{proof}
The encoding of the $[[n,2t-n]]_Q$ quantum code from Lemma~\ref{lm:unit_qmds} is given by
\begin{flalign}
&\hspace{0.2cm}E\left(\ket{s_1}\ket{s_2}\hdots\ket{s_{2t-n}}\ket{0}^{\otimes(2n-2t)}\right)
=\frac{1}{\sqrt{Q^{n-t}}}\sum_{\substack{r_1,r_2,\hdots,\\r_{n-t}\,\in\,\A_Q}}
U_T^{(E)}\!\left(\ket{s_1}\ket{s_2}\hdots\ket{s_{2t-n}}\ket{r_1}\ket{r_2}\hdots\ket{r_{n-t-1}}\ket{r_{n-t}}\right)_{W_T}&\nonumber
\\*[-0.5cm]&\hspace{10cm}\ket{r_1}_{W_{i_1}}\ket{r_2}_{W_{i_2}}\hdots\ket{r_{n-t-1}}_{W_{i_{n-t-1}}}\ket{r_{n-t}}_{W_{i_{n-t}}}\!\!\!\!\!\!&
\label{eq:tmp_1}
\end{flalign}
for all $(s_1,s_2,\hdots,s_{2t-n})\in\A_Q^{2t-n}$.
Take any $j\in T$.
Let $B=(T\setminus\{j\})\cup\{i_{n-t}\}$.
By Lemma~\ref{lm:unit_qmds}, we can now write
\begin{flalign}
&\hspace{0.2cm}E\left(\ket{s_1}\ket{s_2}\hdots\ket{s_{2t-n}}\ket{0}^{\otimes(2n-2t)}\right)
=\frac{1}{\sqrt{Q^{n-t}}}\sum_{\substack{r_1,r_2,\hdots,\\r_{n-t-1},f\\\in\,\A_Q}}
U_B^{(E)}\!\left(\ket{s_1}\ket{s_2}\hdots\ket{s_{2t-n}}\ket{r_1}\ket{r_2}\hdots\ket{r_{n-t-1}}\ket{f}\right)_{W_B}&\nonumber
\\*[-0.75cm]&\hspace{11cm}\ket{r_1}_{W_{i_1}}\ket{r_2}_{W_{i_2}}\hdots\ket{r_{n-t-1}}_{W_{i_{n-t-1}}}\ket{f}_{W_j}\!.\!\!\!\!\!\!&
\label{eq:tmp_2}
\end{flalign}
Note that Eq.~\eqref{eq:tmp_1} and \eqref{eq:tmp_2} give two different expressions for the same state.
Assume that after measuring the $n-t-1$ qudits in systems $W_{i_1},W_{i_2},\hdots,W_{i_{n-t-1}}$ of this state in the computational basis, we obtain the outcome $(s_{2t-n+1},s_{2t-n+2},\hdots,s_{t-1})\in\A_Q^{n-t-1}$.
The new state we obtain after discarding the measured systems is
\begin{eqnarray}
\label{eq:tmp_3}
\frac{1}{\sqrt{Q}}\sum_{r_{n-t}\in\A_Q}
U_T^{(E)}\!\left(\ket{s_1}\ket{s_2}\hdots\ket{s_{t-1}}\ket{r_{n-t}}\right)_{W_T}\,\ket{r_{n-t}}_{W_{n-t}}
=\frac{1}{\sqrt{Q}}\sum_{f\in\A_Q}
U_B^{(E)}\!\left(\ket{s_1}\ket{s_2}\hdots\ket{s_{t-1}}\ket{f}\right)_{W_B}\,\ket{f}_{W_j}
\end{eqnarray}
Note that this state is the same as the encoded state of the $[[t+1,t-1]]_Q$ in Eq.~\eqref{eq:enc_unit_p1_qmds}.
The logical qudits of the $[[t+1,t-1]]_Q$ code can be recovered from the $t$ physical qudits corresponding to $T$ by applying ${U_T^{(E)}}^\dagger$.
From Eq.~\eqref{eq:tmp_3}, it is also clear that the logical qudits of the $[[t+1,t-1]]_Q$ code can be recovered from the $t$ physical qudits corresponding to $B$ by applying ${U_B^{(E)}}^\dagger$.
Hence, the logical qudits of the code can be recovered from any $t$ of its physical qudits.

In other words, the erasure of any one physical qudit can be corrected.
This means the distance of the code is $D'=2$, which leads to equality in the Singleton bound.
\begin{eqnarray}
2(D'-1)=2=(t+1)-(t-1).
\end{eqnarray}
Hence the $[[t+1,t-1]]_Q$ code with the encoding in Eq.~\eqref{eq:enc_unit_p1_qmds} is a quantum MDS code.
\end{proof}

\newtheorem*{orig_lemma_d_geq_t}{Lemma~\ref{lm:d_geq_t}}
\begin{orig_lemma_d_geq_t}
Consider a distributed quantum storage with $n$ nodes using an $[[n,2t-n]]_Q$ quantum MDS code.
To replace a lost node, the replacement node needs quantum communication with at least $t$ of the other $n-1$ nodes.
\end{orig_lemma_d_geq_t}

\begin{proof}
The proof is by contradiction.
Let $e\in[n]$ indicate an erased physical qudit.
Assume that the erased qudit can be replaced after quantum communication with some other physical qudits denoted by $L\subseteq[n]\backslash\{e\}$ where $|L|=t-1$.
Let the system containing the new qudit which replaces the erased qudit be $\widehat{W}_e$.

Consider the maximally mixed state
\begin{equation}
\rho_\mathcal{S}=\frac{1}{\sqrt{Q^{2t-n}}}
\sum_{\ul{s}\in\A_Q^{2t-n}}\ket{\ul{s}}\bra{\ul{s}}.
\end{equation}
in the system $\mathcal{S}$.
Let $\mathcal{R}$ be a reference system such that
\begin{equation}
\ket{\phi}_\mathcal{RS}=\frac{1}{\sqrt{Q^{2t-n}}}\sum_{\ul{s}\in\A_Q^{2t-n}}\ket{\ul{s}}_\mathcal{R}\ket{\ul{s}}_\mathcal{S}.
\end{equation}
Let $\ket{\chi_\text{in}}$ be the state obtained after $\rho_\mathcal{S}$ is encoded by the QMDS code.
The erased qudit in $W_e$ can be replaced with a new physical qudit in $\widehat{W}_e$ using quantum communication with $W_L$ only if there exists a bipartite LOCC$(W_L \widehat{W}_e,$
$\mathcal{R}W_{[n]\setminus L})$ protocol which transforms
\begin{equation}
\ket{\chi_\text{in}}=\frac{1}{\sqrt{Q^t}}\sum_{\ul{s}\in\A_Q^{2t-n}}\ket{\ul{s}}_\mathcal{R}
\sum_{\substack{r_1,r_2,\hdots,r_{n-t}\\\in\A_Q}}U^{(E)}_{L\cup\{e\}}(\ket{\ul{s}}\ket{r_1}\ket{r_2}\hdots\ket{r_{n-t}})_{W_L W_e}
\,\,\,\,(\ket{r_1}\ket{r_2}\hdots\ket{r_{n-t}})_{W{[n]\setminus(L\cup\{e\})}}
\label{eq:d_geq_t_in}
\end{equation}
to the state 
\begin{equation}
\ket{\chi_\text{out}}=\frac{1}{\sqrt{Q^t}}\sum_{\ul{s}\in\A_Q^{2t-n}}\ket{\ul{s}}_\mathcal{R}
\sum_{\substack{r_1,r_2,\hdots,r_{n-t}\\\in\A_Q}}U^{(E)}_{L\cup\{e\}}(\ket{\ul{s}}\ket{r_1}\ket{r_2}\hdots\ket{r_{n-t}})_{W_L \widehat{W}_e}
\,\,\,\,(\ket{r_1}\ket{r_2}\hdots\ket{r_{n-t}})_{W{[n]\setminus(L\cup\{e\})}}
\,\,\,\,\ket{\epsilon}_{W_e}
\label{eq:d_geq_t_out}
\end{equation}
where $\ket{\epsilon}$ is some pure state in $\Hs_Q$.
The states in Eq.~\eqref{eq:d_geq_t_in} and Eq.~\eqref{eq:d_geq_t_out} are due to Lemma~\ref{lm:unit_qmds}.
The existence of such an LOCC protocol implies
\begin{equation}
\sr(W_L \widehat{W}_e)_{\chi_\text{in}}\geq \sr(W_L \widehat{W}_e)_{\chi_\text{out}}
\label{eq:d_geq_t_3}
\end{equation}
due to the Schmidt rank monotonicity under LOCC.
However, it can be shown that
\begin{eqnarray}
&\sr(W_L \widehat{W}_e)_{\chi_\text{in}}=\sr(W_L)_{\chi_\text{in}}=Q^{t-1},&
\\&\sr(W_L \widehat{W}_e)_{\chi_\text{out}}=Q^t&
\end{eqnarray}
which leads to a contradiction.

The above contradiction shows that the erasure of a physical qudit cannot be corrected using quantum communication with only $t-1$ (or fewer) non-erased qudits.
\end{proof}

\newtheorem*{orig_lemma_t_is_suff}{Lemma~\ref{lm:t_is_suff}}
\begin{orig_lemma_t_is_suff}
Consider a distributed quantum storage with $n$ nodes using an $[[n,2t-n]]_Q$ quantum MDS code.
To replace a lost node, it is sufficient for the replacement node to communicate with any $t$ of the other $n-1$ nodes.
\end{orig_lemma_t_is_suff}
\begin{proof}
Let $\in[n]$ indicate the erased qudit and $T\subseteq[n]\{e\}$ indicate any set of $t$ other non-erased qudits.
By Lemma~\ref{lm:unit_qmds}, we can write the encoded state corresponding to any $\ket{\phi}\in\Hs^{\otimes{(2t-n)}}$ as
\begin{eqnarray}
\frac{1}{\sqrt{Q^{n-t}}}
\sum_{\substack{r_1,r_2,\hdots,\\r_{n-t}\in\A_Q}}
U_T^{(E)}\!\left(\ket{\phi}\ket{r_1}\ket{r_2}\hdots\ket{r_{n-t}}\right)_{W_T}
\,\,(\ket{r_1}\ket{r_2}\hdots\ket{r_{n-t-1}})_{W_{[n]\setminus(T\cup\{e\})}}\ket{r_{n-t}}_{W_e}.
\end{eqnarray}
A replacement node $\widehat{W}_e$ can download the $t$ physical qudits from $W_T$ and invert the unitary to obtain
\begin{eqnarray}
\frac{1}{\sqrt{Q^{n-t}}}
\sum_{\substack{r_1,r_2,\hdots,\\r_{n-t}\in\A_Q}}
\left(\ket{\phi}\ket{r_1}\ket{r_2}\hdots\ket{r_{n-t}}\right)_{\widehat{W}_e}
\,\,(\ket{r_1}\ket{r_2}\hdots\ket{r_{n-t-1}})_{W_{[n]\setminus(T\cup\{e\})}}\ket{r_{n-t}}_{W_e}.
\end{eqnarray}
The node $\widehat{W}_e$ then measures its last qudit (and obtains some pure state $\ket{\epsilon}\in\Hs_Q$).
\begin{eqnarray}
\frac{1}{\sqrt{Q^{n-t-1}}}
\sum_{\substack{r_1,r_2,\hdots,\\r_{n-t-1}\in\A_Q}}
\left(\ket{\phi}\ket{r_1}\ket{r_2}\hdots\ket{r_{n-t-1}}\ket{\epsilon}\right)_{\widehat{W}_e}
\,\,(\ket{r_1}\ket{r_2}\hdots\ket{r_{n-t-1}})_{W_{[n]\setminus(T\cup\{e\})}}\ket{\epsilon}_{W_e}.
\end{eqnarray}
It generates a new maximally entangled state using this qudit and an ancilla qudit.
\begin{eqnarray}
\frac{1}{\sqrt{Q^{n-t}}}
\sum_{\substack{r_1,r_2,\hdots,\\r_{n-t-1},\,r\,\in\,\A_Q}}
\left(\ket{\phi}\ket{r_1}\ket{r_2}\hdots\ket{r_{n-t-1}}\ket{r}\ket{r}\right)_{\widehat{W}_e}
\,\,(\ket{r_1}\ket{r_2}\hdots\ket{r_{n-t-1}})_{W_{[n]\setminus(T\cup\{e\})}}\ket{\epsilon}_{W_e}.
\end{eqnarray}
Then it applies the unitary operator $U_T^{(E)}$ on its first $t$ qudits and sends them to their respective nodes in $W_T$\,.
\begin{eqnarray}
\frac{1}{\sqrt{Q^{n-t}}}
\sum_{\substack{r_1,r_2,\hdots,\\r_{n-t-1},\,r\,\in\,\A_Q}}
U_T^{(E)}\!\left(\ket{\phi}\ket{r_1}\ket{r_2}\hdots\ket{r_{n-t-1}}\ket{r}\right)_{W_T}
\,\,(\ket{r_1}\ket{r_2}\hdots\ket{r_{n-t-1}})_{W_{[n]\setminus(T\cup\{e\})}}\ket{r}_{\widehat{W}_e}
\ket{\epsilon}_{W_e}.
\end{eqnarray}
Now the erasure correction is complete with the replacement node communicating with only $t$ non-erased physical qudits.
This process incurs an overall quantum communication of $2t$ qudits and no classical communication.
\end{proof}

\section{Proof for Lemma~\ref{lm:sr_psi_out}}
\label{ap:post_locc_sr}

\begin{lemma}
\label{lm:phi_entr_list}
Consider the state $\ket{\Phi}_{\mathcal{R}W_{[n]}}$ as defined in \eqref{eq:qmds_enc_1}.
For any $A\subseteq[n]$,
\begin{flalign}
&\hspace{5cm}\,\,\mathbf{S}(W_A)_\Phi=
\begin{cases}
\,\,\,|A|\log Q & \,\,\,\,\text{if }1\leq|A|\leq t,
\\[0.05cm]\,\,\,(2t-|A|)\log Q & \,\,\,\,\text{if }t+1\leq|A|\leq n.
\end{cases}\!\!\!\!\!\!&
\end{flalign}
\end{lemma}

\begin{proof}
By Lemma~\ref{lm:unit_qmds}, for any $T\subseteq[n]$ such that $|T|=t$, we can write
\begin{eqnarray}
\ket{\Phi}_{\mathcal{R}W_{[n]}}
=\frac{1}{\sqrt{Q^t}}
\sum_{\ul{s}\in\A_Q^{2t-n}}\sum_{\substack{r_1,r_2,\hdots,\\r_{n-t}\,\in\,\A_Q}}
\ket{\ul{s}}_\mathcal{R}\,U_T^{(E)}\!\left(\ket{\ul{s}}\ket{r_1}\!\hdots\ket{r_{n-t}}\right)_{W_T}
(\ket{r_1}\ket{r_2}\hdots\ket{r_{n-t}})_{W_{[n]\setminus T}}.
\end{eqnarray}
Assume that there are two parties $P_1$ and $P_2$ with $P_1$ holding the subsystem $W_T$ and $P_2$ holding $\mathcal{R}W_{[n]\setminus T}$.
After party $P_1$ locally inverts the unitary operator $U_T^{(E)}$, we obtain
\begin{eqnarray}
\frac{1}{\sqrt{Q^t}}
\sum_{\ul{s}\in\A_Q^{2t-n}}\sum_{\substack{r_1,r_2,\hdots,\\r_{n-t}\,\in\,\A_Q}}
\ket{\ul{s}}_\mathcal{R}\,\left(\ket{\ul{s}}\ket{r_1}\!\hdots\ket{r_{n-t}}\right)_{W_T}
(\ket{r_1}\ket{r_2}\hdots\ket{r_{n-t}})_{W_{[n]\setminus T}}.
\end{eqnarray}
Note that the above state is a maximally entangled state with respect to the bipartition $(W_T,\mathcal{R}W_{[n]\setminus T})$.
This implies that $W_T$ or any of its subsystems containing some of its qudits is in the maximally mixed state.
Hence, for any $A_1\subseteq T$,
\begin{equation}
\mathbf{S}(W_{A_1})=|A_1|\log Q.
\end{equation}
Similarly, $\mathcal{R}W_{[n]\setminus T}$ or any of its subsystems containing some of its qudits is in the maximally mixed state.
Hence, for any $A_2=T\cup B$ where $B\subseteq[n]\setminus T$,
\begin{equation}
\mathbf{S}(W_{A_2})=\mathbf{S}(W_{T\cup B})=\mathbf{S}(\mathcal{R}W_{[n]\setminus(T\cup B)})=(2t-n+n-t-|B|)\log Q=(t-|B|)\log Q=(2t-|A_2|)\log Q.
\end{equation}
\end{proof}

\begin{lemma}
\label{lm:entr_chi}
Let $U_T^{(E)}$ be the unitary operator as defined in Lemma~\ref{lm:unit_qmds}.
Consider the state
\begin{flalign}
\ket{\chi}_{\mathcal{R}W_{[n-1]}}
=\frac{1}{\sqrt{Q^{t-1}}}
\sum_{\ul{s}\in\A_Q^{2t-n}}
\sum_{\substack{r_1,r_2,\hdots,\\r_{n-t-1}\in\A_Q}}
\!\ket{\ul{s}}_\mathcal{R}
\,U_T^{(E)}\!\left(\ket{\ul{s}}\!\ket{r_1}\!\ket{r_2}\hdots\ket{r_{n-t-1}}\!\ket{\epsilon}\right)_{W_T}
\,\ket{r_1}_{W_{t+1}}\!\ket{r_2}_{W_{t+2}}\hdots\ket{r_{n-t-1}}_{W_{n-1}}
\end{flalign}
jointly held by parties $\mathcal{R}, W_1, W_2,\hdots,W_{n-1}$\,.
Here $\ket{\epsilon}$ is some pure state in $\Hs_Q$.
Then, for any $j\in[n-1]$,
\begin{equation}
\sr(W_j)_\chi=Q.
\end{equation}
\end{lemma}
\begin{proof}
Recall from Eq.~\eqref{eq:phi_r_wn} that
\begin{flalign}
\ket{\Phi}_{\mathcal{R}W_{[n]}}
=\frac{1}{\sqrt{Q^t}}
\sum_{\ul{s}\in\A_Q^{2t-n}}
\,\sum_{\substack{r_1,r_2,\hdots,\\r_{n-t}\in\A_Q}}
\ket{\ul{s}}_\mathcal{R}
\,\,U_T^{(E)}\!\left(\ket{\ul{s}}\ket{r_1}\ket{r_2}\hdots\ket{r_{n-t}}\right)_{W_T}
\,\,\,\ket{r_1}_{W_{t+1}}\ket{r_2}_{W_{t+2}}\hdots\ket{r_{n-t}}_{W_n}\!.
\end{flalign}
Now assume that we measure the qudit in $W_n$ in some basis $\{\ket{\mu_j}\}_{j=0}^{Q-1}$ with $\ket{\mu_0}=\ket{\epsilon}$ and obtain the outcome 0.
After this measurement, we obtain the state
\begin{flalign}
&\hspace{0.5cm}\frac{1}{\sqrt{Q^{t-1}}}
\sum_{\ul{s}\in\A_Q^{2t-n}}\sum_{\substack{r_1,r_2,\hdots,\\r_{n-t-1}\in\A_Q}}
\ket{\ul{s}}_\mathcal{R}
\,\,U_T^{(E)}\!\left(\ket{\ul{s}}\ket{r_1}\ket{r_2}\hdots\ket{r_{n-t-1}}\ket{\epsilon}\right)_{W_T}
\,\,\,\ket{r_1}_{W_{t+1}}\ket{r_2}_{W_{t+2}}\hdots\ket{r_{n-t-1}}_{W_{n-1}}\ket{\epsilon}_{W_n}&
\\*[0.2cm]&\hspace{0.8cm}\,\,=\,\,\ket{\chi}_{\mathcal{R}W_{[n-1]}}\ket{\epsilon}_{W_n}.
\end{flalign}
By Lemma~\ref{lm:phi_entr_list}, we know $\mathbf{S}(W_jW_n)_\Phi=2\log Q$.
Hence the subsystem $W_j W_n$, which has two qudits, in $\ket{\Phi}_{\mathcal{R}W_{[n]}}$  is in the maximally mixed state
\begin{equation}
\rho(W_jW_n)_\Phi=\frac{1}{Q^2}\sum_{i,j\,\in\,\A_Q}\ket{i}\ket{\mu_j}\bra{i}\bra{\mu_j}.
\end{equation}
This implies that after the measurement with outcome 0 as described above, the subsystem $W_j$ will be in the state 
\begin{equation}
\rho(W_j)_\chi=\frac{1}{Q}\,\sum_{\,i\in\A_Q}\ket{i}\bra{i}.
\end{equation}
This implies that $\sr(W_j)_\chi=Q$.
\end{proof}

\newtheorem*{orig_lemma_sr_psi_out}{Lemma~\ref{lm:sr_psi_out}}
\begin{orig_lemma_sr_psi_out}
For the state $\ket{\psi_\text{out}}_{W_T W_T'}$ as given in Eq.~\eqref{eq:psi_out},
\begin{eqnarray}
\sr(\widehat{W}_n)_{\psi_\text{out}}&=&Q,
\label{eq:ap_sr_rep_node}
\\\sr(W_J W_J')_{\psi_\text{out}}&\geq&Q^2
\label{eq:ap_sr_psi_out}
\end{eqnarray}
for any $J\subset T$ with $1\leq|J|\leq t-1$.
\end{orig_lemma_sr_psi_out}

\begin{proof}
Recall from Eq.~\eqref{eq:psi_out} that
\begin{flalign}
\ket{\psi_\text{out}}=\frac{1}{\sqrt{Q^t}}\sum_{\ul{s}\in\A_Q^{2t-n}}\sum_{\substack{r_1,r_2,\hdots,\\r_{n-t}\in\A_Q}}
\,U_T^{(E)}\!\left(\ket{\ul{s}}\ket{r_1}\ket{r_2}\hdots\ket{r_{n-t}}\right)_{W_T}
\,\,\ket{r_{n-t}}_{\widehat{W}_n}
\,\,U_T^{(E)}
\big(\!\ket{\ul{s}}\,G_1^\dagger\!\ket{r_1}\hdots\,G_{n-t-1}^\dagger\!\ket{r_{n-t-1}}\ket{\epsilon}\!\big)_{W_T'}.
\nonumber\\*[-0.4cm]
\end{flalign}
The $Q$-dimensional qudit in subsystem $\widehat{W}_n$ for the state $\ket{\psi_\text{out}}$ is in the maximally mixed state.
Hence we obtain Eq.~\eqref{eq:ap_sr_rep_node}.

Without loss of generality, take $T=\{1,2,\hdots,t\}$ and $J=\{1,2,\hdots,j\}$ where $j:=|J|$.
Let $L=T\setminus J$.
Assume that there are two parties $P_1$ and $P_2$ with $P_1$ holding the joint system $W_t W_J'$ and $P_2$ holding the system $W_{[t-1]}W_L'\widehat{W}_n$ of the state $\ket{\psi_\text{out}}$.
Consider the $t+1$ qudits corresponding to $W_T\widehat{W}_n$ in state $\ket{\psi_\text{out}}$.
From Lemma~\ref{lm:unit_p1_qmds}, we know that these qudits form a $[[t+1,t-1]]_Q$ quantum MDS code with distance
\begin{equation}
D'=(t+1)-t+1=2.
\end{equation}
Any $t$ physical qudits of this code can be used to recover its logical qudits.
The party $P_2$ holds $t$ physical qudits of this code corresponding to $W_{[t-1]}\widehat{W}_n$.
After $P_2$ decodes these $t$ qudits to recover the logical qudits of this code, we obtain
\begin{flalign}
\frac{1}{\sqrt{Q^t}}\sum_{\ul{s}\in\A_Q^{2t-n}}
\sum_{\substack{r_1,r_2,\hdots,\\r_{n-t-1},f\in\A_Q}}
\left(\ket{\ul{s}}\ket{r_1}\ket{r_2}\hdots\ket{r_{n-t-1}}\right)_{W_{[t-1]}}
\ket{f}_{\widehat{W}_n}\,\ket{f}_{W_t}
\,U_T^{(E)}
\big(\!\ket{\ul{s}}\,G_1^\dagger\!\ket{r_1}\hdots\,G_{n-t-1}^\dagger\!\ket{r_{n-t-1}}\ket{\epsilon}\!\big)_{W_T'}.
\end{flalign}
Then the party $P_2$ applies the unitary operator $G_i^\dagger$ to the $(2t-n+i)$th qudit in $W_{[t-1]}$ for all $i\in[n-t-1]$ to obtain
\begin{flalign}
&\frac{1}{\sqrt{Q^t}}\sum_{\ul{s}\in\A_Q^{2t-n}}\,
\sum_{\substack{r_1,r_2,\hdots,\\r_{n-t-1},\\f\in\A_Q}}
(\ket{\ul{s}}G_1^\dagger\!\ket{r_1}\hdots G_{n-t-1}^\dagger\!\ket{r_{n-t-1}})_{W_{[t-1]}}
\,\ket{f}_{\widehat{W}_n}\,\ket{f}_{W_j}
\,U_T^{(E)}\big(\!\ket{\ul{s}}\,G_1^\dagger\!\ket{r_1}\hdots\,G_{n-t-1}^\dagger\!\ket{r_{n-t-1}}\ket{\epsilon}\!\big)_{W_T'}\!\!\!\!&\nonumber
\\*&=\frac{1}{\sqrt{Q^t}}\sum_{\ul{s}\in\A_Q^{2t-n}}\,
\sum_{\substack{r_1,r_2,\hdots,\\r_{n-t-1},f\in\A_Q}}
\!\left(\ket{\ul{s}}\ket{r_1}\ket{r_2}\hdots\ket{r_{n-t-1}}\right)_{W_{[t-1]}}
\ket{f}_{\widehat{W}_n}\,\ket{f}_{W_t}
\,U_T^{(E)}\big(\!\ket{\ul{s}}\ket{r_1}\ket{r_2}\hdots\ket{r_{n-t-1}}\ket{\epsilon}\!\big)_{W_T'}&
\\*&=\,\ket{\chi}_{W_{[t-1]}W_T'}\,\ket{Q^+}_{\widehat{W}_n W_t}.
\label{eq:psi_out_modified}
\end{flalign}
Since local unitary operations do not change the Schmidt rank, the Schmidt rank of the subsystem held by $P_1$ \textit{i.e.} $W_t W_J'$ for the state in Eq.~\eqref{eq:psi_out_modified} is the same as $\sr(P_1)_{\psi_\text{out}}$.
Hence,
\begin{eqnarray}
\sr(W_tW_J')_{\psi_\text{out}}
\,\,=\,\,\sr(W_t)_{Q^+}\sr(W_J')_\chi
\,\,=\,\,Q\,\sr(W_J')_\chi
\,\,=\,\,Q^{|J|+1}.
\end{eqnarray}
The last equality in the above equation is due to Lemma~\ref{lm:entr_chi}.
\end{proof}

\end{document}